\documentclass{article}

\usepackage[english]{babel}

\usepackage[left=1.3in,right=1.3in,top=1.2in,bottom=1.2in]{geometry}

\usepackage{amsmath}
\usepackage{amssymb}
\usepackage{amsfonts}
\usepackage{graphicx}

\usepackage[round]{natbib}
\usepackage{soul}
\usepackage{url}
\usepackage[hidelinks]{hyperref}
\usepackage[utf8]{inputenc}
\usepackage{amsthm}
\usepackage{booktabs}
\usepackage{algorithm}
\usepackage{bbm}

\newtheorem{theorem}{Theorem}

\newtheorem{corollary}{Corollary}
\newtheorem{example}{Example}

\newtheorem{proposition}{Proposition}
\newtheorem{definition}{Definition}

\usepackage{mathtools}
\usepackage{todonotes}

\usepackage{dsfont}

\newcommand{\naturals}{{{\mathbb{N}}}}

\newcommand{\CC}{\mathcal{C}}
\newcommand{\G}{\mathcal{G}}

\newcommand{\ie}{i.e.,\xspace}
\newcommand{\eg}{e.g.,\xspace}

\newcommand{\col}{m}
\newcommand{\ind}{\mathds{1}}
\newcommand{\indmod}{\ind((\col-1)\nmid(q-1))}

\usepackage{pgf, tikz} 
\usetikzlibrary{arrows, automata} 
\usepackage{cleveref}
\usepackage{multirow}

\usepackage{diagbox}

\usepackage{tabularx}

\usepackage{pifont}
\newcommand{\no}{\color{red}{\ding{55}}}

\allowdisplaybreaks

\usepackage{graphicx}
\usepackage{pgfplots}

\usepackage{authblk}

\usepackage[ruled,vlined,linesnumbered,algo2e]{algorithm2e}

\usepackage[font={footnotesize},labelfont={sc}]{caption} 

\pgfplotsset{compat=1.18}

\title{Quantifying Core Stability Relaxations in Hedonic Games\footnote{A previous version of this paper appeared under the title ``Relaxed Core Stability for Hedonic Games with Size-Dependent Utilities'' at the 48th International Symposium on Mathematical Foundations of Computer Science (MFCS 2023).}}
\author[1]{Tom Demeulemeester}
\author[2]{Jannik Peters}
\affil[1]{Department of Economics, Faculty of Business and Economics (HEC), University of Lausanne}
\affil[2]{School of Computing, National University of Singapore}

\begin{document}
	\maketitle
	
	\begin{abstract}
		We study relationships between different relaxed notions of core stability in hedonic games, which are a class of coalition formation games. Our unified approach applies to a newly introduced family of hedonic games, called $\alpha$-hedonic games, which contains previously studied variants such as fractional and additively separable hedonic games. In particular, we derive an upper bound on the maximum factor with which a blocking coalition of a certain size can improve upon an outcome in which no deviating coalition of size at most $q$ exists. Counterintuitively, we show that larger blocking coalitions might sometimes have lower improvement factors. We discuss the tightness conditions of our bound, as well as its implications on the price of anarchy of core relaxations. Our general result has direct implications for several well-studied classes of hedonic games, allowing us to prove two open conjectures by \citet{FMM21a} for fractional hedonic games. 
	\end{abstract}
	
	\section{Core Stability in Hedonic Games}
	Coalition formation is one of the core topics of multiagent systems and algorithmic game theory. Hedonic games \citep{DrGr80a} constitute one of the most popular subcases of coalition formation. In a hedonic game, the goal is to divide a set of agents having preferences over subsets of the other agents into disjoint coalitions respecting these preferences. Over the years, multiple different ways of representing the agents' preferences and multiple different solution concepts emerged. Among the strongest solution concepts is core stability \citep{BoJa02a}: a coalition structure is core stable, if no subset of agents could together form a new coalition in which they are all better off than in the original coalition structure. While being a seemingly natural concept, core stable outcomes may fail to exist, even for very simple preference structures \citep{ABB+19a, ABS13a}.
	In these structures, it is also often computationally intractable to decide whether a core stable outcome exists \citep{SuDi10a, PeEl15a}.\footnote{For some preference structures, it can even be hard to find a coalition structure in the core, even if it is guaranteed to exist, see for instance the paper of \citet{BuKo21a}.} One standard approach to deal with these impossibilities is through approximation, see for instance \citet{CMS24b} or \citet{CJMW20a} for recent examples in clustering and multiwinner voting.
	For hedonic games, \citet{FMM21a} were the first to introduce two natural approximate notions of core stability for hedonic games: (i) $q$-size (core) stability, which requires that no blocking coalitions of size at most $q$ exist, (ii) $k$-improvement (core) stability, which requires that no blocking coalition exists in which every agent improves by a factor of more than $k$. They further applied these notions to the special case of fractional hedonic games \citep{ABB+19a} and showed that a $2$-size stable outcome and a $2$-improvement stable outcome always exist. Fractional hedonic games, however, constitute only one of many possible special cases of possible hedonic games. 
	
	In this paper, we contribute to this literature in three ways. First, we propose a new family of hedonic games, called $\alpha$-hedonic games, in which the utility an agent receives from being in a coalition of size $\col$ is equal to the sum of the cardinal utilities it ascribes to the other agents in that coalition, multiplied by a factor $\alpha_\col$ which depends on the size of that coalition. Several well-studied classes of hedonic games, such as fractional hedonic games \citep{ABB+19a}, modified fractional hedonic games \citep{Olse12a}, and additively separable hedonic games \citep{BoJa02a} are a special case of $\alpha$-hedonic games. While these classes have mostly been studied independently in the literature, the unified techniques that we propose expose the inherent structural similarities between them. Additionally, our unified approach facilitates the study of core stability in previously unstudied variants of $\alpha$-hedonic games, as illustrated in \Cref{sec:novel}.
	
	Second, we further study the relationship between the two relaxations of core stability that were introduced by \citet{FMM21a}. Our main result quantifies, for any $\alpha$-hedonic game and for any approximately $q$-size stable outcome, the maximum factor with which the agents can improve their utility by forming a blocking coalition of size $\col\geq q+1$. As a corollary, this allows us to prove two conjectures by \citet{FMM21a}: (i) every $q$-size stable outcome is $\frac{q}{q-1}$-improvement stable for fractional hedonic games and (ii) the $q$-size core price of anarchy (i.e., the worst-case approximation to the social welfare of any $q$-size core stable outcome) is exactly $\frac{2q}{q-1}$ for fractional hedonic games. 
	
	Third, we introduce the use of mathematical programming techniques to gain insights into hedonic games. To help us identify the rather counterintuitive bound in our main result, we formulated an integer linear program that computed instances allowing for ``extreme" blocking coalitions in which all agents improved by a given factor. Simultaneously, the generated instances can be used to illustrate the tightness of our results.
	
	\paragraph*{Related Work}
	Since its inception, hedonic games have been a widely studied topic in algorithmic game theory, with several works studying axiomatic or computational properties of hedonic games. For an overview on earlier developments, we refer the reader to the book chapter by~\citet{AzSa15a}. In recent years, several new models and optimality notions for hedonic games have been introduced and analyzed. Among the most popular of these models are the aforementioned fractional hedonic games, introduced by \citet{ABB+19a}, and studied in various forms by, e.g., \citet{BFF+18a}, \citet{AGG+15b}, or \citet{CMM19a}. Fractional hedonic games are also related to the model of hedonic diversity games \citep{BEI19a, BoEl20a, GHK+23a} in which agents possess types and derive utility based on the fraction of agents of their own type in their coalition.
	
	The paper closest to ours is the work by~\citet{FMM21a}, who introduced the aforementioned notions of $q$-size and $k$-improvement core stability for fractional hedonic games. Alternative simplifications of core stability were introduced by \citet{CMM19a}, who studied a local variant of core stability for simple fractional hedonic games, i.e., hedonic games in which all utility values are either $0$ or $1$. In their local variant of core stability, the agents deviating are required to form a clique. For this weakened notion, they show that core stable outcomes always exist and can be computed via improving response dynamics.
	
	Hedonic games in general are a widely studied topic in both computer science and operations research. Earlier work has focused on the computational complexity of computing or verifying stable coalition structures with both tasks frequently being computationally intractable \citep{GaSa19a, SuDi10a}. Stability notions from hedonic games have recently found application in the study of fair clustering, with \citet{AAK+22a} and \citet{ACL+23a} studying a variant of individual stability and modified fractional hedonic games, while \citet{CMS24b} study core stability in a non-centroid fair clustering setting. The main difference between the clustering and hedonic games setting is that for clustering problems the number of clusters is typically fixed, while for hedonic games any arbitrary number of coalitions can be formed. In the context of bandit algorithms, \citet{CPM23a} used an approach based on hedonic games to cluster users for a recommendation task. \citet{AGIM24a} study a hedonic games setting in which the participating agents are sharing their resources with the other agents in their coalition. In this setting, \citeauthor{AGIM24a} show that, in this setting, the existence of stable coalition structures can be guaranteed if the sharing rule used in coalitions follows their principle of ``solidarity''.

	\section{Preliminaries}
	For any $n \in \mathbb{N}^+$ and $\alpha\colon \mathbb{N}_{\ge 1} \to \mathbb{R}_{>0}$ an \emph{$\alpha$-hedonic game} ($\alpha$HG) consists of a set of \emph{agents} $A = \{a_1, \dots, a_n\}$ with a \emph{utility function} $u\colon A \times A \to \mathbb{R}$. We restrict ourselves to \emph{symmetric} $\alpha$-hedonic games (S-$\alpha$HGs) in this paper, and require that $u(i,j) = u(j,i)$ for all $i,j \in A$. A \emph{coalition} is a subset of $A$ and a \emph{coalition structure} is a partition of $A$ into coalitions. The utility of an agent $i$ in a coalition $C$ is $u_i(C) \coloneqq \sum_{j \in C} \alpha(\lvert C \rvert) \cdot u(i,j)$. We assume that $u(i,i) = 0$. For a coalition structure $\mathcal C$ the utility $u_i(\mathcal{C})$ of the coalition structure for agent $i$ is the utility of the coalition agent $i$ belongs to. To simplify notation, for agents $a_i$ and $a_j$ and $C\subseteq A$ we also write $u_i(j) \coloneqq u(i, j)$ and $u_i(C) \coloneqq u_{a_i}(C)$.
	
	The utility that an agent experiences from a coalition of size $\col$ in an $\alpha$HG is the sum of the utilities for the other agents in that coalition, weighted by a factor $\alpha(\col)$, which only depends on the size of the coalition. As such, the class of $\alpha$-hedonic games generalizes multiple previously studied hedonic game classes, e.g.,
	\begin{itemize}
		\item \emph{Symmetric additively separable hedonic games} (S-ASHGs) with $\alpha(\col) = 1$ for any $\col \in \mathbb{N}$.
		\item \emph{Symmetric fractional hedonic games} (S-FHGs) with $\alpha(\col) = \frac{1}{\col}$ for any $\col \in \mathbb{N}$.
		\item \emph{Symmetric modified fractional hedonic games} (S-MFHGs) with $\alpha(\col) = \frac{1}{\col-1}$ for any $\col \in \mathbb{N}^+$ and $\alpha(1) = 0$.
	\end{itemize}
	We only study \emph{symmetric} hedonic games (see \Cref{ex:intro}). To improve readability, we omit the prefix ``S-" in the abbreviations of the different hedonic game classes in the remainder of the paper.
	
	A given coalition structure $\mathcal C$ is 
	\begin{itemize}
		\item \emph{core stable} if for any coalition $C$ it holds that $u_i(C) \le u_i(\mathcal{C})$ for at least one $i \in C$;
		\item \emph{$q$-size core stable} if for any coalition $C$ with $\lvert C \rvert \le q$ it holds that $u_i(C) \le u_i(\mathcal{C})$ for at least one $i \in C$;
		\item  \emph{$k$-improvement core stable} if for any coalition $C$ it holds that $u_i(C) \le k u_i(\mathcal{C})$ for at least one $i \in C$;
		\item  \emph{$(q,k)$-core stable} if for any coalition $C$ with $\lvert C \rvert = q$ it holds that $u_i(C) \le k u_i(\mathcal{C})$ for at least one $i \in C$.
	\end{itemize}
	Note that $q$-size core stability is equivalent to being $(\col,1)$-core stable for $1\leq \col\leq q$. Similarly, a coalition is $k$-improvement stable if and only if it is $(q,k)$-core stable for all values of $q$. If there is a coalition witnessing a violation to one of these criteria, for instance a coalition $C$ with $u_i(C) > u_i(\mathcal{C})$ for all $i \in C$, we say that $C$ is a \emph{blocking coalition}. In further parts of the paper, we shorten $\alpha(\col)$ to $\alpha_\col$ to increase readability. Given two integers~$a$ and~$b$, we denote that~$a$ is divisible by~$b$ as $a\mid b$, and that~$a$ is not divisible by~$b$ as $a\nmid b$. Moreover, denote by~$\ind\colon\{\text{true},\text{false}\}\to\{0,1\}$ the indicator function such that for any boolean condition $x$
	\begin{equation*}
		\ind(x) = \begin{cases}
			1 &\text{if } x \text{ is true},\\
			0 &\text{if } x \text{ is false}.
		\end{cases}
	\end{equation*}
	Lastly, every $\alpha$HG can be represented by a graph $G(A,E,w)$, where $A$ represents the set of agents, and $E$ contains an undirected edge $\{i,j\}$ between agents $i$ and $j$ with weight $w_{ij} = u(i,j) = u(j,i)$ if $u(i,j)>0$. Alternatively, given a coalition $C\subseteq A$, we denote the subgraph of $G(A,E,w)$ that is induced by only considering the agents in $C$ by $G(C)$. 
	
	Before turning to our result, we present two simple examples, with the second example motivating why we exclusively focus on symmetric instances.
	\begin{figure}[h]
		\centering
		\begin{tikzpicture}[
			> = stealth, 
			shorten > = 1pt, 
			auto,
			node distance = 1.5cm, 
			semithick 
			]
			\tikzstyle{every state}=[
			draw = black,
			fill = white,
			minimum size = 4mm
			]
			\node[state] (a4) at (0,3) {$a_4$};
			\node[state] (a3) at (0,0) {$a_3$};
			\node[state] (a2) at (3,0) {$a_2$};
			\node[state] (a1) at (3,3) {$a_1$};
			
			\path[-] (a1) edge node {$3$} (a2);
			\path[-] (a2) edge  node {$3$} (a3);
			\path[-] (a3) edge  node {$3$} (a4);
			\path[-] (a4) edge  node {$3$} (a1);
			\path[-] (a3) edge  node {$2$} (a1);
		\end{tikzpicture}
		\hspace{1cm}
		\begin{tikzpicture}[
			> = stealth, 
			shorten > = 1pt, 
			auto,
			node distance = 1.5cm, 
			semithick 
			]
			\tikzstyle{every state}=[
			draw = black,
			fill = white,
			minimum size = 4mm
			]
			
			\node[state] (a3) at (0,0) {$a_3$};
			\node[state] (a2) at (3,0) {$a_2$};
			\node[state] (a1) at (1.5,2) {$a_1$};
			
			\path[->] (a1) edge [bend left] node {$M$} (a2);
			\path[->] (a2) edge [bend left] node {$M$} (a3);
			\path[->] (a3) edge [bend left] node {$M$} (a1);
			
			\path[->] (a2) edge node {$0$} (a1);
			\path[->] (a3) edge node {$0$} (a2);
			\path[->] (a1) edge node {$0$} (a3);
			\node[] at (-0.85, -0.75) {$u_3(\CC)=1$};
			\node[] at (3.75, -0.75) {$u_2(\CC)=1$};
			\node[] at (1.45, 2.65) {$u_1(\CC)=1$};
		\end{tikzpicture}
		\caption{Example of a symmetric hedonic game on the left and of a blocking coalition for an asymmetric game on the left, for which the improvement ratio is unbounded.}
		\label{fig:ex}
	\end{figure}
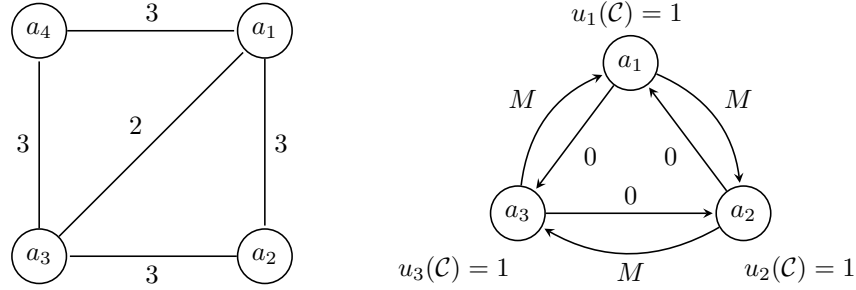
	\begin{example}
		\label{ex:intro}
		First, consider the hedonic game induced by the graph on the left of Figure \ref{fig:ex}, with utilities as indicated by the edges and omitted edges indicating a utility of $0$. Here, consider the coalition structure $\{\{a_1, a_2\}, \{a_3, a_4\}\}$. In an ASHG (i.e., additively separable), the utility of every agent would be $3$, and the coalition structure would be $2$-size core stable, but not $3$-size core stable, since, for instance, $\{a_1, a_2, a_3\}$ would block. Further, the coalition structure is $(3, \frac{5}{3})$- and $(4, 2)$-core stable and thus also $2$-improvement core stable.
		In a FHG, on the other hand, the utility of every agent would be $\frac{3}{2}$ and the coalition $\{a_1, a_2, a_3\}$ would still block. The coalition consisting of all agents, however, would no longer be blocking, since the utility of agent $a_2$ would be $\frac{6}{4} = \frac{3}{2}$, which was their utility in the original coalition structure. 
		Finally, in a MFHG, the utility of every agent would be $3$ and the coalition structure would be core stable. Even the coalition $\{a_1, a_2, a_3\}$ would no longer block, since the utility of agent $a_1$ would be $\frac{5}{2} < 3$. 
		
		Secondly, to motivate the choice of symmetric hedonic games, consider the (asymmetric) hedonic game depicted on the right of \Cref{fig:ex} with all three agents originally being in a coalition structure $\CC$ in which they experience utility $1$. This coalition structure would be $2$-size core stable. However, there is no upper bound on the improvement ratio for the coalition consisting of all three agents, as $M$ goes to infinity. We note that this behaviour can be observed independently of the considered $\alpha$ function. 
	\end{example}
	
	\subsection{Our results}
	~\citet{FMM21a} conjectured that for fractional hedonic games, every $q$-size core stable coalition structure is also $\frac{q}{q-1}$-improvement core stable. We refine this conjecture and show that every $q$-size core stable coalition structure is \begin{align}\left(\col, 1 + \frac{\left\lfloor \frac{1}{q-1} (\col-2)\right\rfloor}{\col}\right)\text{-core stable} 
		\label{eq:frac_bound}\end{align} for any $\col \ge q + 1$. As $1 + \frac{\left\lfloor \frac{1}{q-1} (\col-2)\right\rfloor}{\col} \le \frac{q}{q-1}$ for any $\col$ this implies the conjecture of \citeauthor{FMM21a}. Further, this result together with the results of \citeauthor{FMM21a} also allows us to confirm their second conjecture that the \emph{price of anarchy} of $q$-size stability is exactly $\frac{2q}{q-1}$.
	
	To gain a better intuition of this quite unhandy term, we refer the reader to \Cref{ta:bounds} and the left panel of \Cref{fig:values}.
	In particular, note that for $q > 2$ this bound is not monotone. We provide some intuition behind this non-monotonicity in \Cref{sec:intuition}. 
	
	\begin{table}[ht]
		
		\begin{center}
			\caption{For the given combinations of $q$ and $\col$, the table contains the value $f(q,\col)$ derived from \Cref{eq:frac_bound}, such that a $q$-size core stable is $(\col, f(q,\col))$-stable in FHGs. \label{ta:bounds}}
			\begin{tabular}{l c c c c c c c c |c} \toprule
				\backslashbox{$q$}{$\col$}
				& $3$ &$4$& $5$ & $6$ & $7$ & $8$ & $9$ & $\ldots$ & $\frac{q}{q-1}$ \\\midrule
				$2$ & $\frac{4}{3}$ & $\frac{6}{4}$ & $\frac{8}{5}$ & $\frac{10}{6}$ & $\frac{12}{7}$ & $\frac{14}{8}$ & $\frac{16}{9}$ &
				&
				$2$\\
				\rule[-1.5ex]{0pt}{5ex}$3$ & -- & $\frac{5}{4}$ & $\frac{6}{5}$ & $\frac{8}{6}$ & $\frac{9}{7}$ & $\frac{11}{8}$ & $\frac{12}{9}$ & 
				$\ldots$ &
				$\frac{3}{2}$\\
				\rule[-1.5ex]{0pt}{5ex}$4$ & -- & -- & $\frac{6}{5}$ & $\frac{7}{6}$ & $\frac{8}{7}$ & $\frac{10}{8}$ & $\frac{11}{9}$ &
				&
				$\frac{4}{3}$\\
				\bottomrule
			\end{tabular}

		\end{center}
	\end{table}
	
	\begin{figure}[t]
		\centering
		\includegraphics[width=0.75\linewidth]{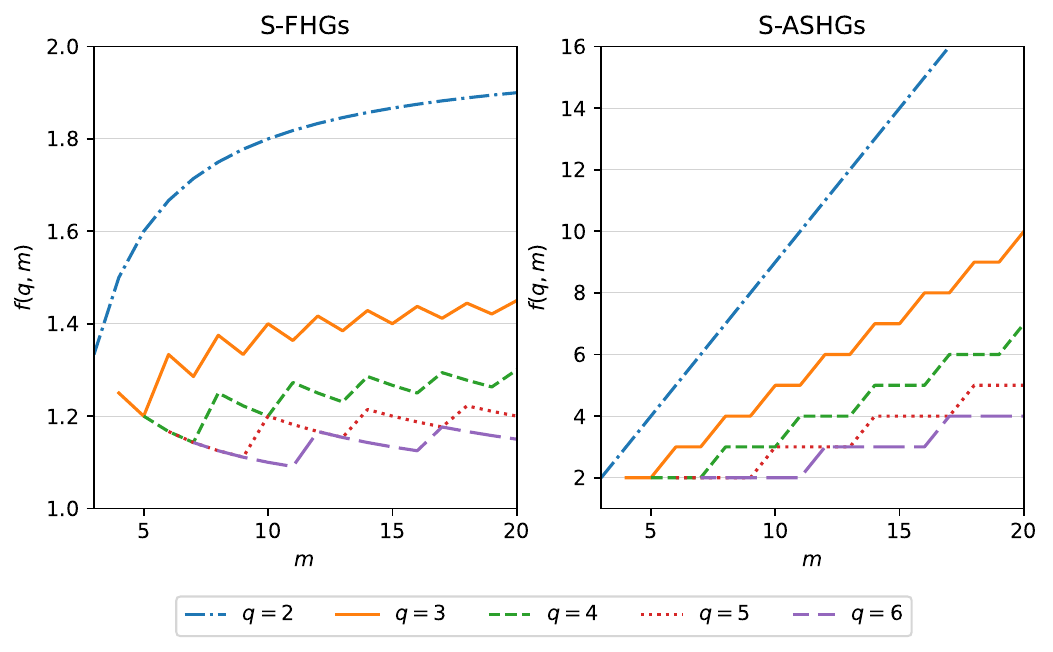}
		\caption{Plotted values of $f(q,\col)$ for FHGs (left) and ASHGs (right) such that every $q$-size core stable coalition structure is $(\col,f(q,\col))$-core stable.}
		\label{fig:values}
	\end{figure}

	In fact, our proof does not only apply to FHGs, but to all symmetric $\alpha$-hedonic games. The more general result that we are able to show is that every coalition structure in an $\alpha$HG that is $k$-improvement stable for coalitions of size $1,\ldots,q$ is also $(\col, f(q,\col))$-core stable with 
	\begin{align*}
		&f(q,\col) = \max\left(1,k\left(\left\lfloor\frac{\col-1}{q-1}\right\rfloor\frac{\alpha_\col}{\alpha_q}+\frac{\indmod\alpha_\col}{\alpha\left((\col-1)\bmod(q-1)+1\right)}\right)\right).
	\end{align*} 
	As we discuss in Section \ref{sec:applications}, this bound is equivalent to \Cref{eq:frac_bound} for FHGs when $k=1$. For ASHGs this implies a bound where 
	\(    
	f(q,\col) = 1+\left\lfloor\frac{\col-2}{q-1}\right\rfloor
	\) when $k=1$, as illustrated in the right panel of \Cref{fig:values}. Moreover, for MFHGs, this implies the result by Monaco et al.\ \cite[][Theorem 14]{MMV20a} that a core stable coalition structure always exists.
	
	Further, in \Cref{sec:lowerBounds} we derive lower bounds on these values as well, and show tightness for various combinations of $(q,\col)$, and for various types of hedonic games. A summary of our results can be found in \Cref{ta:summary}.
	
	To justify the study of $q$-size stability, consider the following preliminary result, which guarantees the existence of a 2-size stable coalition structure in $\alpha$HGs for any possible $\alpha$ function.
	
	\begin{proposition}
		\label{prop:2_size_exists}
		There exists a 2-size stable coalition structure in every $\alpha$HG.
	\end{proposition}
	\begin{proof}
		The proof follows a similar idea as the proof of Theorem 14 by \citet{MMV20a} to show that the core of MFHGs is not empty. Consider an $\alpha$HG represented by a graph $G(A,E,w)$. Denote by $G^0(A,E^0,w)$ the modified graph in which all edges with negative weights are removed, \ie $E^0=\{e\in E: w_e\geq 0\}$. Define the coalition structure $\CC$ by greedily pairing the unmatched agents with the highest weights in $G^0$ (breaking ties arbitrarily), and by leaving the remaining agents unmatched. Note the two agents $i,j\in A$ cannot be matched if $\{i,j\}\notin E^0$. By definition of a greedy matching, $\CC$ is not strictly blocked by a coalition of size 2. Further, no matched agent strictly prefers to stay unmatched: For every agent~$i\in A$, being unmatched yields utility $\alpha_1 u(i,i) = 0$, while being matched yields a non-negative utility as $\alpha_2 \geq 0$ and $w_e\geq 0$ for all $e\in E^0$. 
	\end{proof}
	
	\section{Methodology}
	\label{sec:method}
	As displayed in \Cref{fig:values}, the bound in our main result can be rather counterintuitive: larger blocking coalitions might experience a smaller improvement factor for several $\alpha$-hedonic games, including FHGs. To help us identify the bound in our main result, we have formulated an integer linear program. This formulation generated instances that contained a $q$-size stable coalition structure, and that allowed for blocking coalitions of size $\col$ that were $(\col, \gamma)$-core stable, with $\col\geq q+1$. For FHGs, for example, \Cref{ta:bounds} displays the largest values of $\gamma$ for which this integer linear program managed to find feasible solutions. Note that for larger values of $\gamma$, our integer linear program could often not prove infeasibility.
	
	Assume we are given a $q$-size stable coalition structure $\CC$ in an $\alpha$HG in which each agent $i\in A$ experiences utility $u_i(\CC)$. Our goal is to evaluate whether utilities exist such that we can identify a blocking coalition $C\subseteq A$ with $|C|\geq q+1$ in which all agents improve by at least a factor $\gamma$ with respect to $\CC$. Let auxiliary variables $y_{S,i}$ be one if agent $i$ in subset $S\subseteq C$, with $|S|\leq q$, does not improve with respect to $\CC$, and zero otherwise. Then, the coalition $C$ is $(|C|,\gamma)$-blocking if and only if the following system of linear equations is feasible.
	
	\begin{align}
		&&\sum_{j\in S\setminus\{i\}}u(i,j) - \frac{u_i(\CC)}{\alpha_{|S|}} &\leq M (1-y_{S,i}) &\forall S\subseteq C:|S|\leq q, \forall i\in S\label{eq:IP1}\\
		&&\sum_{i\in S}y_{S,i} &\geq 1 \qquad &\forall S\subseteq C: |S|\leq q\label{eq:IP2}\\
		&&u(i,j) &= u(j,i) \qquad &\forall\{i,j\}\subseteq C\label{eq:IP3}\\
		& &\sum_{j\in C\setminus\{i\}}u(i,j) &> \frac{\gamma u_i(\CC)}{\alpha_{|C|}} \qquad &\forall i\in C\label{eq:IP4}\\
		&&y_{S,i}&\in\{0,1\} \qquad &\forall  S\subseteq C:|S|\leq q, \forall i\in S\label{eq:IP5}\\
		&&u(i,j) &\in \mathbb{R} &\forall \{i,j\}\subseteq C\label{eq:IP6}\\
		&&u_i(\CC) &\in \mathbb{R} &\forall i\in C \label{eq:IP7}
	\end{align}
	Inequalities (\ref{eq:IP1}) and (\ref{eq:IP2}) ensure that the choice of utilities does not violate the fact that $\CC$ is $q$-size core stable, where $M$ represents a large number. Jointly, inequalities (\ref{eq:IP1}) and (\ref{eq:IP2}) enforce that in each subset $S\subseteq C$ of size at most $q$, there is at least one agent who does not improve with respect to $\CC$. While constraints (\ref{eq:IP3}) enforce the resulting $\alpha$-hedonic game to remain symmetric, constraints (\ref{eq:IP4}) impose that all agents in $C$ improve by a factor of at least $\gamma$, and that it is therefore a $(|C|,\gamma)$-blocking coalition. 
	
	We have used this formulation not only to obtain the bound in the main result, but also to guide the tightness proofs of this bound in \Cref{sec:lowerBounds}. For the cases when a theoretical lower bound could not be established, the generated instances when setting the parameters $|C|=\col$ and $\gamma = f(q,\col)$ could illustrate the tightness of our result (\Cref{ta:additional_values}). Note that while strict inequalities are not allowed by linear solvers, this can be resolved by deducting a small value on one side of the inequality. The corresponding C\texttt{++} code that solves this formulation using Gurobi, together with the tight instances in \Cref{ta:additional_values}, is available online at \url{https://github.com/DemeulemeesterT/Relaxations-Core-Stability-Alpha-Hedonic-Games}.

	\section{Main result}
	\label{sec:main}
	We begin with our main result, which quantifies the relationship between the two considered relaxed notions of core stability in symmetric $\alpha$-hedonic games.
	\begin{theorem}
		\label{th:alphaHG}
		Any coalition structure in an $\alpha$HG that is $k$-improvement core stable for coalitions of size $1,\ldots,q$ is also $(\col, f(q,\col))$-core stable, for any integers $\col,q$ with $\col\geq q+1$, and $f(q,\col)=$
		\begin{equation*}\max\left(1,k\left(\left\lfloor\frac{\col-1}{q-1}\right\rfloor\frac{\alpha_\col}{\alpha_q}+\frac{\indmod\alpha_\col}{\alpha\left((\col-1)\bmod(q-1)+1\right)}\right)\right).
		\end{equation*}
	\end{theorem}
	
	\begin{proof}
		
		Consider a coalition structure~$\mathcal{C}$ which is $k$-improvement core stable for coalitions of size $1,\ldots,q$ and a blocking coalition~$C$ of size~$\col\geq q+1$. For a given coalition $C'\subseteq C$ we let $$w(C') = \sum_{a_i \in C'} \left(u_i(C') - 2 \frac{k\alpha_{q-1}}{\alpha_q} u_i(\CC)\right)$$ denote the \emph{modified social welfare} of coalition $C'$. Let $\mathcal{C}_{q-1}$ be the set of coalitions of size $q-1$ and consider the weighted hypergraph $(C, \mathcal{C}_{q-1},w)$. Let $M = \{C_1, \dots, C_{\left\lfloor\frac{\col-1}{q-1}\right\rfloor}\}$ be any maximum weight hypergraph matching with regard to $w$ in this hypergraph, i.e., selection of non-overlapping sets from $\CC_{q-1}$ of size $\left\lfloor\frac{\col-1}{q-1}\right\rfloor$ with maximum weight. We note that a maximum weight hypergraph matching not of size $\left\lfloor\frac{\col-1}{q-1}\right\rfloor$ might have a larger weight, due to $w$ being potentially negative. Let $C_0$ be the set of agents unmatched by $M$. The goal of our proof is now to show that there must be an unmatched agent who can only improve by a factor of at most $$\max\left(1,k\left(\left\lfloor\frac{\col-1}{q-1}\right\rfloor\frac{\alpha_\col}{\alpha_q}+\frac{\indmod\alpha_\col}{\alpha\left((\col-1)\bmod(q-1)+1\right)}\right)\right).$$ The set $C_0$ contains exactly $(\col-1)\bmod(q-1)+1$ agents that are unmatched by $M$. Let $a_0\in C_0$. For any $i \in [\left\lfloor\frac{\col-1}{q-1}\right\rfloor]$ we know that not all of the agents in coalition $\{a_0\} \cup C_i$ cannot improve by a factor of more than $k$ (as this coalition is of size at most $q$). Thus, either one of the two following cases has to hold:
		
		\begin{enumerate}
			\item[(i)] $\sum_{a_j \in C_i} u(0,j) \le  \frac{ku_0(\CC)}{\alpha_q}$,
			\item[(ii)] $\sum_{a_j \in C_i} u(0,j) > \frac{ku_0(\CC)}{\alpha_q}$ and there is an $a_\ell\in C_i$ with $\sum_{a_j \in C_i \cup \{a_0\}} u(\ell,j) \le \frac{ku_\ell(\CC)}{\alpha_q}$.
		\end{enumerate}
		
		By definition, the modified welfare of coalition $C_i \cup \{a_0\} \setminus \{a_\ell\}$ equals
		\begin{align*}
			w(C_i \cup \{a_0\} \setminus \{a_\ell\})&=\sum_{a_j \in C_i \setminus \{a_\ell\}} \left(u_j(C_i \cup \{a_0\} \setminus \{a_\ell\}) - 2\frac{k\alpha_{q-1}}{\alpha_q}u_j(\CC)\right)\\
			&\qquad+ \sum_{a_j \in C_i \setminus \{a_\ell\}} \alpha_{q-1}u(0,j) - 2\frac{k\alpha_{q-1}}{\alpha_q}u_0(\CC).
		\end{align*}
		By specifying the term $u_j(C_i \cup \{a_0\} \setminus \{a_\ell\})$ and some rearranging, this is equivalent to
		\begin{align*}
			&\sum_{a_j \in C_i \setminus \{a_\ell\}} \biggl(u_j(C_i) + \alpha_{q-1}u(j,0) - \alpha_{q-1}u(j,\ell) - 2\frac{k\alpha_{q-1}}{\alpha_q}u_j(\CC)\biggr)\\
			&\qquad + \sum_{a_j \in C_i \setminus \{a_\ell\}} \alpha_{q-1}u(0,j)-  2\frac{k\alpha_{q-1}}{\alpha_q}u_0(\CC)\\
			&= \sum_{a_j \in C_i \setminus \{a_\ell\}} \left(u_j(C_i)  - 2\frac{k\alpha_{q-1}}{\alpha_q}u_j(\CC)\right)  - \sum_{a_j \in C_i \setminus \{a_\ell\}} \alpha_{q-1}u(j,\ell)\\
			&\qquad + \sum_{a_j \in C_i \setminus \{a_\ell\}} 2\alpha_{q-1}u(0,j)  - 2\frac{k\alpha_{q-1}}{\alpha_q}u_0(\CC). 
		\end{align*}
		
		We will show that case (ii) cannot hold by contradiction. To do so, we will show that $$w(C_i \cup \{a_0\} \setminus \{a_\ell\}) > w(C_i),$$ which implies that the hypergraph matching cannot have been of maximum weight in case (ii). Applying the first condition of case (ii) to the last term yields
		\begin{align*}
			w(C_i \cup \{a_0\} \setminus \{a_\ell\})&> \sum_{a_j \in C_i \setminus \{a_\ell\}} \left(u_j(C_i)  - 2\frac{k\alpha_{q-1}}{\alpha_q}u_j(\CC)\right)- \sum_{a_j \in C_i} \alpha_{q-1}u(\ell,j)\\&\qquad + \sum_{a_j \in C_i \setminus \{a_\ell\}} 2\alpha_{q-1}u(0,j) - \sum_{a_j \in C_i} 2\alpha_{q-1}u(0,j),
		\end{align*}
		where we used the symmetry of the utilities and the fact that $u(\ell, \ell) = 0$ to rewrite the second summation. By removing the double-counting in the last two summations for all agents except $a_\ell$, this reduces to
		\begin{align*}
			w(C_i \cup \{a_0\} \setminus \{a_\ell\})&>\sum_{a_j \in C_i \setminus \{a_\ell\}} \left(u_j(C_i)  - 2\frac{k\alpha_{q-1}}{\alpha_q}u_j(\CC)\right) - \sum_{a_j \in C_i} \alpha_{q-1}u(\ell,j) \\&\qquad - 2 \alpha_{q-1}u(0,\ell).
		\end{align*}
		Adding and deducting the term $\sum_{a_j \in C_i} 2\alpha_{q-1}u(\ell,j)$ allows us to rewrite the right hand side of this expression as
		\begin{align*}
			w(C_i \cup \{a_0\} \setminus \{a_\ell\})&> \sum_{a_j \in C_i \setminus \{a_\ell\}} \left(u_j(C_i)  - 2\frac{k\alpha_{q-1}}{\alpha_q}u_j(\CC)\right) - \sum_{a_j \in C_i} \alpha_{q-1}u(\ell,j) \\&\qquad  + \sum_{a_j \in C_i} 2\alpha_{q-1}u(\ell,j) -\sum_{a_j \in C_i} 2\alpha_{q-1}u(\ell,j) - 2 \alpha_{q-1}u(0,\ell)\\
			&= \sum_{a_j \in C_i \setminus \{a_\ell\}} \left(u_j(C_i)  - 2\frac{k\alpha_{q-1}}{\alpha_q}u_j(\CC)\right)+ \sum_{a_j \in C_i} \alpha_{q-1}u(\ell,j)\\&\qquad - \sum_{a_j \in C_i \cup \{ a_0\}} 2\alpha_{q-1}u(\ell,j).
		\end{align*}
		Lastly,  applying the second condition in case (ii) to the last summation yields the desired contradiction, implying that the hypergraph matching was not of maximum weight in case (ii):
		\begin{align*}
			w(C_i \cup \{a_0\} \setminus \{a_\ell\})&>\sum_{a_j \in C_i \setminus \{a_\ell\}} \left(u_j(C_i)  - 2\frac{k\alpha_{q-1}}{k\alpha_q}u_j(\CC)\right)+ \sum_{a_j \in C_i} \alpha_{q-1}u(\ell,j) \\ &\qquad -  2\frac{k\alpha_{q-1}}{\alpha_q}u_\ell(\CC)\\
			& = w(C_i).
		\end{align*}
		Because we obtain a contradiction in case (ii), case (i) should hold. As a result, $\sum_{a_j \in C_i} u(0,j) \le  \frac{ku_0(\CC)}{\alpha_q}$ holds for every $C_i$ and every $a_0 \in C_0$. The result follows in each of the following two cases, depending on the size of $C_0$.
		
		In the case that $|C_0|=1$, we can apply the inequality in case (i) to bound the utility of $a_0$ in blocking coalition $C$ by 
		\begin{align*}
			u_0(C)=\sum_{a_j \in C} \alpha_{\col}u(0,j) &= \sum_{C_i \in M} \sum_{a_j \in C_i} \alpha_{\col}u(0,j)\\ &\le \sum_{C_i \in M} \frac{k\alpha_\col}{\alpha_q} u_0(\CC)= k\left(\frac{\col-1}{q-1}\frac{\alpha_\col}{\alpha_q}u_0(\CC)\right).
		\end{align*}
		Hence, agent $a_0$ does not improve by more than a factor of $k\left(\frac{\col-1}{q-1}\frac{\alpha_\col}{\alpha_q}\right)$, which implies the result, because $(q-1)$ divides $(\col-1)$ when $|C_0|=1$.
		
		In the case that $|C_0|>1$, there has to be at least one agent $a_0$ in $C_0$ with $\sum_{a_i \in C_0} u(0,i) \le \frac{ku_0(\CC)}{\alpha(\lvert C_0 \rvert)} = \frac{ku_0(\CC)}{\alpha((\col-1) \bmod (q-1)+1)}$, since the set $C_0$ of unmatched agents is non-blocking. Applying this observation and the inequality in case (i), we obtain that
		\begin{align*}
			&\sum_{a_j \in C} \alpha_{\col}u(0,j) = \sum_{C_i \in M} \sum_{a_j \in C_i} \alpha_{\col}u(0,j) + \sum_{a_i \in C_0} \alpha_{\col}u(0,i) \\ 
			&\le \sum_{C_i \in M} \frac{k\alpha_\col}{\alpha_q}u_0(\CC) + \frac{k\alpha_\col}{\alpha\left((\col-1)\bmod(q-1)+1\right)}u_0(\CC) \\
			&= k\left(\left\lfloor\frac{\col-1}{q-1}\right\rfloor\frac{\alpha_\col}{\alpha_q} + \frac{\alpha_\col}{\alpha((\col-1)\bmod(q-1)+1)}\right)u_0(\CC),
		\end{align*}
		which concludes the proof.
	\end{proof}
	Note that setting $k=1$ in the statement of \Cref{th:alphaHG} implies a result for $q$-size stable coalition structures in $\alpha$HGs.
	\begin{corollary}
		\label{cor:alphaHG}
		Any $q$-size core stable coalition structure in an $\alpha$HG is also $(\col, f(q,\col))$-core stable, for any integers $\col,q$ with $\col\geq q+1$, and $f(q,\col)=$
		\begin{equation*}\max\left(1,\left\lfloor\frac{\col-1}{q-1}\right\rfloor\frac{\alpha_\col}{\alpha_q}+\frac{\indmod\alpha_\col}{\alpha\left((\col-1)\bmod(q-1)+1\right)}\right).
		\end{equation*}
	\end{corollary}
	
	
	\section{Applications for Existing Hedonic Games Models}
	\label{sec:applications}
	
	In this section, we will show the general applicability of \Cref{th:alphaHG} for various previously studied models of hedonic games that belong to the class of $\alpha$-hedonic games.
	\subsection{Fractional Hedonic Games}
	In the case of fractional hedonic games, which are $\alpha$HGs in which $\alpha_\col = \frac{1}{\col}$, we can simplify the bound in \Cref{th:alphaHG} as follows.
	\begin{corollary}
		Any coalition structure in a FHG that is $k$-improvement core stable for coalitions of size $1,\ldots,q$ is also $\left(\col, k\left(1 + \frac{\left\lfloor \frac{1}{q-1} (\col-2)\right\rfloor}{\col}\right)\right)$-core stable, for any integers $\col,q$ with $\col\geq q+1$.
		\label{cor:FHG}
	\end{corollary}
	
	\begin{proof} In the case that $(m-1)$ is divisible by $(q-1)$, it holds that
		
		\begin{align*}
			&\left\lfloor\frac{\col-1}{q-1}\right\rfloor\frac{\alpha_\col}{\alpha_q}+\frac{\indmod\alpha_\col}{\alpha\left((\col-1)\bmod(q-1)+1\right)} = \frac{\col-1}{q-1}\frac{q}{\col}\\
			&= \frac{\col(q-1)+\col-q}{(q-1)\col} = 1 + \frac{\frac{1}{q-1} (\col-q)}{\col}=1 + \frac{\left\lfloor \frac{1}{q-1} (\col-2)\right\rfloor}{\col}.
		\end{align*}
		
		
		\noindent In the case that $(\col-1)\bmod(q-1)\neq0$, we obtain that 
		\begin{align*}
			&\left\lfloor\frac{\col-1}{q-1}\right\rfloor\frac{\alpha_\col}{\alpha_q}+\frac{\indmod\alpha_\col}{\alpha\left((\col-1)\bmod(q-1)+1\right)} = \left\lfloor\frac{\col-1}{q-1}\right\rfloor\frac{q}{\col} + \frac{(\col - 1) \bmod (q-1) + 1}{\col}.
		\end{align*}
		By definition of the floor operator, we can rewrite this expression as
		\begin{align*}
			&\frac{(m-1)-(m-1) \bmod (q-1)}{q-1} \frac{q}{\col} +\frac{(\col - 1) \bmod (q-1) + 1}{\col}\\
			&= \frac{(m-1)-(m-1) \bmod (q-1)}{q-1} \frac{q-1}{\col} +\frac{(\col-1)-(\col-1) \bmod (q-1)}{q-1}\frac{1}{\col}\\
			&\qquad  +\frac{(\col - 1) \bmod (q-1) + 1}{\col} \\
			&= \frac{1}{\col} \left(\col + \frac{(\col-1)-(\col-1) \bmod (q-1)}{q-1}\right) = 1 + \frac{\left\lfloor \frac{1}{q-1} (\col-2)\right\rfloor}{\col},
		\end{align*}
		where the last step holds since $(\col - 1) \bmod (q-1) > 0$. Note that this expression is not smaller than one, and we can therefore discard the maximum operator in $f(q,\col)$. Multiplying by $k$ gives us the desired bound.
	\end{proof}
	
	As a corollary we obtain an answer to the conjecture of \citet{FMM21a}, by confirming that every $q$-size core stable outcome is also $\frac{q}{q-1}$-improvement stable.
	\begin{corollary}
		\label{cor:general_FHG}
		For FHGs, every $q$-size core stable coalition structure is $\frac{q}{q-1}$-improvement stable.
	\end{corollary}
	\begin{proof}
		This result follows from the observation that\begin{align*}
			&1+\frac{\left\lfloor \frac{1}{q-1} (\col-2)\right\rfloor}{\col} \le 1+\frac{ \frac{1}{q-1} (\col-2)}{\col} \le  1 + \frac{1}{q-1} = \frac{q}{q-1}\end{align*}
		holds for all $m$. Thus, there is no coalition of size $m \geq q+1$ in which every agent improves by a factor of more than $\frac{q}{q-1}$.
	\end{proof}
	
	If we restrict ourselves to the case of \emph{simple symmetric fractional hedonic games (SS-FHG)}, i.e., symmetric FHGs with binary utilities, we can show that this bound is not tight. This proof further strengthens Theorem 2 of \citet{FMM21a}.
	\begin{theorem}
		For every simple symmetric fractional hedonic game, any 3-size core stable coalition structure is $$\left(\col, \frac{3}{2}\frac{(\col-1)}{\col}\right)\text{-core stable,}$$ for any integer~$\col\geq 4$.
	\end{theorem}
	\begin{proof}
		We only prove the result for even values of $\col$, as the result for odd $\col$ follows from the general result of \Cref{th:alphaHG}. Consider a 3-size core stable coalition structure~$\CC$ and a blocking coalition~$C$ of size $\col\geq4$. Further, we assume that every agent in $C$ improves by more than a factor of $k$. Because all agents experience a strict improvement by forming~$C$, each agent should be adjacent to at least one edge in the related graph~$G(C)$. Since 3-size stability implies 2-size stability, there should be at least one agent~$a_i\in C$ for whom $u_i(\CC)\geq \frac{1}{2}$. Since $a_i$ improves by more than a factor of $k$ it holds that
		\begin{equation}
			\label{eq:min_edges}
			\sum_{a_\ell\in C\setminus\{a_i\}}u(i,\ell) > k\cdot \col \cdot u_i(\CC) \geq \frac{\col}{2},
		\end{equation}
		which implies that $a_i$ should be adjacent to at least $\frac{\col}{2}+1$ edges in $G(C)$. Denote the set of agents that are connected to $a_i$ through these edges by $C'\subset C$. For any triplet of agents $\{a_i,a_j,a_k\}$ with $\{a_j,a_k\}\subset C'$, the 3-size core stability of~$\CC$ implies that either $u_i(\CC) \geq \frac{2}{3}$, $u_j(\CC) \geq \frac{1}{2}$, or $u_k(\CC) \geq\frac{1}{2}$ must hold. If $u_i(\CC) \geq \frac{2}{3}$ holds, then because $a_i$ improves in $C$ by a factor of more than $k$ we get that 
		\begin{equation}
			\label{eq:UB_k}
			\frac{2}{3}k \le k \cdot u_i(\CC) < \frac{1}{\col}\sum_{a_\ell\in C\setminus\{a_i\}}u(i,\ell) \le \frac{\col-1}{\col}.
		\end{equation}
		Thus, we get that $k < \frac{3}{2} \frac{\col-1}{\col}$ in this case.
		Alternatively, assume without loss of generality that $u_j(\CC)\geq\frac{1}{2}$. Following the logic from \Cref{eq:min_edges}, $a_j$ should be adjacent to at least $\frac{\col}{2}+1$ edges in $G(C)$. This implies that there exists an agent $a_\ell\in C'$ such that the agents $\{a_i,a_j,a_\ell\}$ form a triangle, and hence at least one of these three agents should experience a utility of at least $\frac{2}{3}$ in~$\CC$. Thus, Equation~(\ref{eq:UB_k}) holds for this agent and the result follows.
	\end{proof}
	
	\subsection{Modified Fractional Hedonic Games}
	In the case of modified fractional hedonic games, where $\alpha_\col= \frac{1}{\col-1}$, our result implies the non-emptiness of the core that was first shown by Monaco et al.\ \citep[][Theorem 14]{MMV20a}. In fact, we can use our result to derive a simple condition to show the non-emptiness of the core for a wide variety of $\alpha$-hedonic games.
	
	\begin{corollary}
		Any $\alpha$HG with $\frac{(m - 1) \alpha_m}{\alpha_2} \le 1$ always admits a core stable outcome.
		\label{cor:core_char}
	\end{corollary}
	\begin{proof}
		From \Cref{prop:2_size_exists} we know that every $\alpha$HG admits a $2$-size stable coalition structure. Thus, we get that 
		\begin{align*}
			f(2,\col) &= \left\lfloor\frac{\col-1}{1}\right\rfloor\frac{\alpha_\col}{\alpha_2} \le 1
		\end{align*}
		and therefore this $2$-size stable coalition structure is also core stable.
	\end{proof}
	
	Note that MFHGs are one of the transitional classes of $\alpha$-hedonic games for which the condition in \Cref{cor:core_char} holds in equality. 
	
	\begin{corollary}{(Monaco et al.\ \citep[][Theorem 14]{MMV20a})}
		\label{cor:MFSG} A core-stable coalition structure always exists in MFHGs. 
	\end{corollary}

	\subsection{Additively Separable Hedonic Games}
	Finally, for additively separable hedonic games, in which $\alpha_\col=1$ for all values of $\col$, \Cref{th:alphaHG} implies the following result.
	\begin{corollary}
		Any coalition structure in an ASHG that is $k$-improvement core stable for coalitions of size $1,\ldots,q$ is also $\left(\col, k\left(1+\left\lfloor\frac{\col-2}{q-1}\right\rfloor\right)\right)$-core stable, for any integers $\col,q$ with $\col\geq q+1$.
		\label{cor:ASHG}
	\end{corollary}

	\section{Intuition}
	
	We now try to give some intuition for our main result and the reason it can be non-monotone. We will do so by constructing blocking coalitions in FHGs where all agents can improve exactly with the factor obtained in our main result. 
	
	Consider a $3$-size core stable coalition structure~$\CC$ in a FHG where $u_i(\CC)=1$ holds for all agents, and a blocking coalition~$C$ of size $\col$. Since 3-size stability implies 2-size stability, every edge in $G(C)$ has a weight of at most two. Further, we note that the edges of weight two cannot form a triangle, as the agents in that triangle would otherwise all get a utility of $\frac{4}{3} > 1$, witnessing a violation to 3-size stability. As a result, all edges that do not have weight two can have a weight of at most one. Assume, for illustration purposes, that we restrict ourselves to integer weights, and that all edges have weights of either two or one. Then, 3-size stability implies that the edges of weight two should form a triangle-free graph.
	
	By Mantel's theorem, we know that any triangle-free graph with $m$ nodes has at most $\left\lfloor \frac{m^2}{4}\right\rfloor$ edges. \Cref{exp:triangle} depicts maximum-size triangle-free graphs with $4$ to $7$ vertices. Even when distributing the edges of weight two as evenly as possible across the nodes, by the pigeonhole principle there must be an agent who is adjacent to at most $\left\lfloor\frac{2\left\lfloor \frac{m^2}{4}\right\rfloor}{m} \right\rfloor = \left\lfloor \frac{m}{2}\right\rfloor$ edges of weight two.  By considering the remaining edges of weight one, this agent's utility is at most 
	$$\frac{1}{m} \left(2\left\lfloor \frac{m}{2}\right\rfloor + \left(m - 1- \left\lfloor \frac{m}{2}\right\rfloor\right)\right) = 1+\frac{\left\lfloor \frac{m}{2}\right\rfloor-1}{m} = 1+\frac{\left\lfloor \frac{m-2}{2}\right\rfloor}{m},$$
	which corresponds to the non-monotone bound obtained in \Cref{cor:FHG} when $q=3$. 
	
	As a concrete instance, \Cref{exp:triangle} shows that, when $m=4$, every agent can be adjacent to exactly two edges of weight two. Likewise, when $m = 5$, there must exist an agent who is adjacent to at most two edges of weight two. The utility of this agent in the entire coalition, however, will be lower than in the coalition of size four as their average utility decreases. Thus, this agent would improve less in a coalition with five agents than in the one with four.
	
	A similar reasoning can be used to obtain the bound for 3-size stable coalition structures in ASHGs, by simply replacing all edge weights of two by one, and all edge weights of one by zero.
	
		
		
	
	\label{sec:intuition}
	
	\begin{figure}
		\centering
		\begin{tikzpicture}[
			> = stealth, 
			shorten > = 0pt, 
			auto,
			node distance = 1.5cm, 
			semithick 
			]
			\tikzstyle{every state}=[
			draw = black,
			fill = white,
			minimum size = 4mm
			]
			\node[above, font=\bfseries] at (1,2.7) {$\mathbf{m=4}$};
			\node[state] (a1) at (0,2) {$1$};
			\node[state] (a4) at (0,0) {$4$};
			\node[state] (a3) at (2,0) {$3$};
			\node[state] (a2) at (2,2) {$2$};
			
			\path[-] (a1) edge node {}  (a2);
			\path[-] (a2) edge  node  {}(a3);
			\path[-] (a3) edge  node {}(a4);
			\path[-] (a4) edge  node {} (a1);
		\end{tikzpicture}
		\hspace{0.8cm}
		\begin{tikzpicture}[
			> = stealth, 
			shorten > = 0pt, 
			auto,
			node distance = 1.5cm, 
			semithick 
			]
			\tikzstyle{every state}=[
			draw = black,
			fill = white,
			minimum size = 4mm
			] 
			\node[above, font=\bfseries] at (0,1.7) {$\mathbf{m=5}$};
			\node[state] (A) at (0,1.236) {1};
			\node[state] (B) at (-1.177,0.382) {$5$};
			\node[state] (C) at (-0.726,-1.000) {4};
			\node[state] (D) at (0.726,-1.000) {3};
			\node[state] (E) at (1.177,0.382) {2};
			\draw (A) -- (B);
			\draw (B) -- (C);
			\draw (C) -- (E);
			\draw (D) -- (E);
			\draw (E) -- (A);
			\draw (D) -- (B);
		\end{tikzpicture}
		\hspace{0.8cm}
		\begin{tikzpicture}[
			> = stealth, 
			shorten > = 0pt, 
			auto,
			node distance = 1.5cm, 
			semithick 
			]
			\tikzstyle{every state}=[
			draw = black,
			fill = white,
			minimum size = 4mm
			] 
			\node[above, font=\bfseries] at (0,1.7) {$\mathbf{m=6}$};
			\node[state] (A) at (0.66666667, 1) {2};
			\node[state] (B) at (1.154700538,0) {3};
			\node[state] (C) at (0.66666667, -1) {4};
			\node[state] (D) at (-0.66666667, -1) {5};
			\node[state] (E) at (-1.154700538,0) {6};
			\node[state] (F) at (-0.6666667, 1) {1};
			\draw (A) -- (B);
			\draw (B) -- (C);
			\draw (C) -- (D);
			\draw (D) -- (E);
			\draw (E) -- (F);
			\draw (F) -- (A);
			\draw (F) -- (C);
			\draw (A) -- (D);
			\draw (B) -- (E);
		\end{tikzpicture}
		\hspace{0.8cm}
		\begin{tikzpicture}[
			> = stealth, 
			shorten > = 0pt, 
			auto,
			node distance = 1.5cm, 
			semithick 
			]
			\tikzstyle{every state}=[
			draw = black,
			fill = white,
			minimum size = 4mm
			] 
			\node[above, font=\bfseries] at (0,1.6) {$\mathbf{m=7}$};
			\node[state] (G) at (0, 1.236) {1};
			\node[state] (A) at (-0.9663437123304849, 0.7706333950973947) {7};
			\node[state] (B) at (-1.205010899456734,-0.27503587437000454
			) {6};
			\node[state] (C) at (-0.536280301549302,-1.1135975207273898
			) {5};
			\node[state] (D) at (0.5362803015493017,-1.11359752072739) {4};
			\node[state] (E) at (1.205010899456734,-0.2750358743700048
			) {3};
			\node[state] (F) at (0.966343712330485,0.7706333950973945
			) {2};
			
			\draw (A) -- (B);
			\draw (B) -- (C);
			\draw (D) -- (E);
			\draw (E) -- (F);
			\draw (F) -- (G);
			\draw (G) -- (A);
			
			\draw (G) -- (C);
			\draw (G) -- (D);
			\draw (E) -- (A);
			\draw (B) -- (D);
			\draw (B) -- (F);
			\draw (G) -- (C);
			\draw (E) -- (C);
			\draw (D) -- (G);
		\end{tikzpicture}
		\caption{Examples of maximum-size triangle-free graphs with $4$ to $7$ vertices.}
		\label{exp:triangle}
	\end{figure}
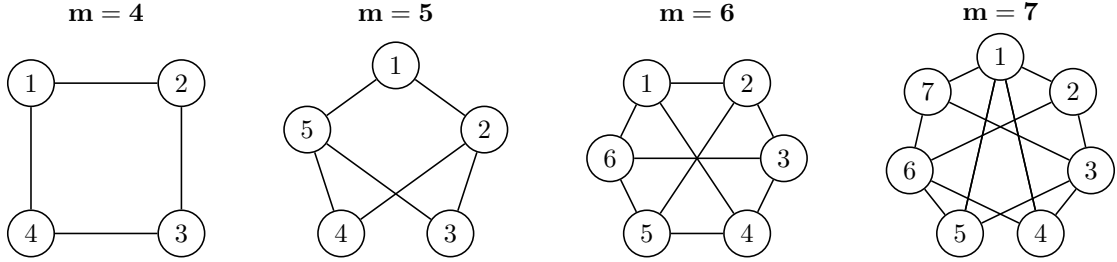

	\section{Lower Bounds}
	\label{sec:lowerBounds}
	In this section, we will illustrate the tightness of the bound in \Cref{th:alphaHG}. We will focus on the setting where $k=1$ in the formulation of \Cref{th:alphaHG}, and will therefore provide lower bounds on the improvement factor of agents in $q$-size stable coalition structures (\Cref{cor:alphaHG}). Note that any lower bound for $k = 1$ can easily be modified to a lower bound for a larger factor $k'$ by multiplying every utility by $k'$. 
	\subsection{Hospitable Hedonic Games}
	\label{subsec:hospit}
	To construct a lower bound for a large class of hedonic games we consider the following subclass of $\alpha$HGs.
	\begin{definition}
		\label{def:hospitable} 
		A function $\alpha \colon [n] \to \mathbb{R}$ is \emph{hospitable} if $\frac{\alpha_q}{\alpha_{q-1}} \geq \frac{q-2}{q-1}$ for all integers $q \ge 2$. Accordingly, an $\alpha$HG is hospitable if $\alpha$ is hospitable.
	\end{definition}
	
	The intuition behind hospitable $\alpha$HGs is that the utility of an agent in a coalition of size $q-1$ will never decrease with more than a factor $\frac{q-2}{q-1}$ when an additional agent is added to that coalition. The class of hospitable $\alpha$HGs is rather large, and includes the previously studied models of ASHGs, FHGs, and MFHGs. As discussed in \Cref{sec:novel}, where we illustrate how our results can be applied to novel hedonic games models, hospitability is a natural requirement that is commonly satisfied by other variants.
	
	First, we can show that the bound derived in \Cref{th:alphaHG} is tight for hospitable $\alpha$HGs when $(\col - 1) \bmod (q-1) = 0$.
	
	\begin{theorem}
		For any hospitable $\alpha$, there exists an instance of an $\alpha$HG that contains a $q$-size core stable coalition structure
		which is not $\left(\col,\delta\right)$-core stable for any $\delta<\frac{\alpha_\col(\col-1)}{\alpha_q(q-1)}$, with $q,\col\in \naturals$ and $\col\geq q+1$.
		\label{thrm:mod_tight}
	\end{theorem}
	\begin{proof}
		We construct an instance of an $\alpha$HG that is $q$-size core stable, but which allows for a blocking coalition of size $\col\geq q+1$ in which all agents improve with at least a factor $\left(\frac{\alpha_\col(\col-1)}{\alpha_q(q-1)}\right)$. Given a coalition structure~$\mathcal{C}$ in which $u_i(\mathcal{C})=1$ for all agents~$i$, and a blocking coalition~$C$ of size~$\col$, let $u(i,j)=\frac{1}{\alpha_q(q-1)}$ for all agent pairs $\{i,j\}\subset C$. The resulting $\alpha$HG is $q$-size core stable, since for any coalition $C'$ of size at most $q$ the utility of agent~$i$ in $C'$ is at most $\frac{(\lvert C'\rvert - 1) \alpha_{\lvert C'\rvert}}{\alpha_q(q-1)} \le \frac{\alpha_q(q-1)}{\alpha_q(q-1)} = 1$, 
		where the inequality is implied by recursively applying the definition of a hospitable $\alpha$HG. As $u_i(\mathcal{C})=1$ and since the utility of agent $i$ in $C$ is $\frac{\alpha_\col (\col-1)}{\alpha_q (q-1)}$ for all agents $i\in \mathcal{C}$, this implies that the coalition structure~$\mathcal{C}$ is not $\left(\col,\delta\right)$-core stable for any $\delta<\frac{\alpha_\col(\col-1)}{\alpha_q(q-1)}$.
	\end{proof}
	
	Since if $(\col - 1) \bmod (q-1) = 0$ it holds that $\lfloor \frac{\col-1}{q-1} \rfloor = \frac{\col-1}{q-1}$ and $\indmod\alpha_\col = 0$, this implies that the bound obtained in \Cref{th:alphaHG} is tight for $(\col - 1) \bmod (q-1) = 0$ and thus also for $q = 2$. Figure \ref{fig:LB_FHG} illustrates the tightness of this lower bound for FHGs. Further, we show tightness of the bound in \Cref{th:alphaHG} for hospitable $\alpha$HGs when $q = 3$.

	\begin{theorem}
		\label{theorem:LB_3}
		For any hospitable $\alpha$, there exists an instance of an $\alpha$HG that contains a $3$-size core stable coalition structure~$\CC$ which is not $\left(\col,\delta\right)$-core stable for any $\delta<f(3,\col)$, with $q,\col\in \naturals$ and $\col\geq 4$.
	\end{theorem}
	\begin{proof}
		Note that when $f(3,\col)=1$, the result follows directly. Moreover, when $\col$ is odd, so $(\col-1) \bmod 2 = 0$, the result follows from \Cref{thrm:mod_tight}. When $\col$ is even and when $f(3,\col)>1$, we first observe that 
		\begin{align*}
			f(3,\col) = \left\lfloor\frac{\col-1}{2}\right\rfloor\frac{\alpha_\col}{\alpha_3}+\frac{\ind\left((\col-1)\bmod 2\right)\alpha_\col}{\alpha\left((\col-1)\bmod 2+1\right)} 
			= \frac{\col-2}{2}\frac{\alpha_\col}{\alpha_3}+\frac{\alpha_\col}{\alpha_2}.
		\end{align*}
		Now we assume that we are given $\col$ agents $c_1, \dots, c_\col$ with $u_i(\CC) = 1$ for each $c_i$. We partition the agents into two sets $C_1, C_2$ with $C_1 = \{c_1, \dots, c_{\frac{\col}{2}}\}$ and $C_2 = \{c_{\frac{\col}{2} + 1}, \dots, c_\col\}$. 
		If two agents $c_i$ and $c_j$ are in the same set, we define $u(i,j) = \frac{1}{\alpha_3}  - \frac{1}{\alpha_2}$. Otherwise, we set $u(i,j) = \frac{1}{ \alpha_2}$.
		
		For any two agents $\{c_i, c_j\}$ it thus holds that $u_i(\{c_i, c_j\}) \le \alpha_2 \max(\frac{1}{\alpha_3}  - \frac{1}{\alpha_2}, \frac{1}{\alpha_2})= \max(\frac{\alpha_2}{\alpha_3} - 1, 1) \le 1$ and, therefore, these two agents do not form a blocking coalition. Further, for any three agents $\{c_i, c_j, c_k\}$ we have two cases: (i) either all three agents come from the same set, then we get that $u_i(\{c_i, c_j, c_k\}) = \frac{2\alpha_3}{\alpha_3}  - \frac{2\alpha_3}{\alpha_2} = 2 - \frac{2\alpha_3}{\alpha_2} \le 2 - \frac{\alpha_3}{\alpha_3} = 1$; (ii) one agent (without loss of generality $c_k$) has to be from a different partition than the other two; then it holds that $u_i(\{c_i, c_j, c_k\}) = \frac{\alpha_3}{\alpha_3}  - \frac{\alpha_3}{\alpha_2} + \frac{\alpha_3}{\alpha_2} = 1$.
		Hence, this coalition is $3$-stable. 
		
		Finally, we get that \begin{align*}
			u_i(\{c_1, \dots, c_\col\}) = \alpha_\col \left(\frac{\col-2}{2} \left(\frac{1}{\alpha_3} - \frac{1}{\alpha_2}\right) + \frac{\col}{2} \frac{1}{\alpha_2}\right)= \frac{\col-2}{2}\frac{\alpha_\col}{\alpha_3}+\frac{\alpha_\col}{\alpha_2} = f(3,\col).
		\end{align*}
	\end{proof}
	
	\Cref{subsec:odd_even} provides an example of a non-hospitable $\alpha$HG where the bound in \Cref{cor:alphaHG} is not tight.

	\subsection{Fractional and Additively Separable Hedonic Games}
	Lastly, we provide additional tightness results of the bound in \Cref{cor:alphaHG} for FHGs and ASHGs.  We defer the proofs of Theorems~\ref{theorem:LB_4} and \ref{theor:ASHG_4} to Appendices~\ref{ap:proof_FHG_4} and~\ref{ap:proof_ASHG_4}.
	\begin{theorem}
		For both FHGs and ASHGs there exist instances in which a $q$-size core stable coalition structure~$\CC$ exists which is not $\left(q+1, \delta\right)$-core stable
		\begin{itemize}
			\item  for any $\delta < \frac{q+2}{q+1}$ for FHGs;
			\item  for any $\delta < 2$ for ASHGs.
		\end{itemize}
		
		\label{theor:q_q+1}
	\end{theorem}
	\begin{proof}
		First, for  FHGs assume we are given $q+1$ agents $a_1, \ldots, a_{q+1}$ with $u_i(\CC)=1$. Let the edge weights be such that the edges with weight two form a cycle and all other edges have weight one, i.e., let $u(i,j) = 2$ for all $(i,j)\subset \CC$ for which $j=i+1$, let $u(n,1) = 2$, and let $u(k,l)=1$ for all other edges. Note that in each subset $C\subset\CC$ with $\vert C\vert < q+1$ there is at least one agent who is adjacent to at most one edge of weight two to the other agents in~$C$, because the edges with weight two form a cycle over all $q+1$ agents. Hence, for each subset $C\subset \CC$ with $\vert C\vert < q+1$ there is at least one agent with a utility of at most $\frac{2+\vert C\vert-2}{\vert C\vert}=1$, which implies that $\CC$ is $q$-size core stable. Furthermore, the coalition of all $q+1$ agents offers a utility of $\frac{2\cdot 2 + q-2}{q+1} = \frac{q+2}{q+1}$. As a result, $\CC$ is not $\left(q+1, \frac{q+2}{q+1}-\varepsilon\right)$-core stable for any $\varepsilon>0$.
		
		To adapt this to ASHGs, replace all edge weights of one by zero, and all weights of two by one. Now in the entire coalition, everyone would get a utility of two while in any smaller subset there must exist a voter getting a utility of at most one.
	\end{proof}
	
	\begin{theorem}
		\label{theorem:LB_4}
		There exists a 4-size core stable coalition structure $\mathcal{C}$ in a FHG which is not $\left(\col,\delta\right)$-core stable, for any $\delta < 1+\frac{\left\lfloor\frac{1}{3}(\col-2)\right\rfloor}{\col}$, and for any integer $\col\geq 5$.
	\end{theorem}
	
	\begin{theorem}
		\label{theor:ASHG_4}
		There exists a 4-size core stable coalition structure~$\CC$ in an ASHG which is not $(\col,\delta)$-core stable, for any $\delta < 1+\left\lfloor\frac{\col-2}{3} \right\rfloor$, and for any integer $\col\geq5$.
	\end{theorem}
	
	\subsection{Summary Lower Bounds}
	\label{subsec:summ}
	\begin{table}[t]
		\caption{Additional values of $(q,\col)$ for which we found instances proving the tightness of the bound in \Cref{cor:alphaHG} by using the integer linear programming approach described in \Cref{sec:method}.\label{ta:additional_values}}
		\centering
		\begin{tabular}{ll}
			\toprule
			\multirow{2}{*}{\textbf{FHG}} & $(5, \col\leq 8)$, $(6, \col\leq 10)$, $(7, \col\leq 10)$,\\
			& $(8, \col\leq 11)$\\ &\\
			\multirow{3}{*}{\textbf{ASHG}} & $(5,\col\leq 8)$, $(6,\col\leq 10)$, $(7, \col\leq12)$,\\
			& $(8, \col\leq13)$, $(9,\col\leq 13)$, $(10, \col\leq 13)$,\\
			& $(11,\col\leq 13)$, $(12,13)$\\
			\bottomrule
		\end{tabular}
	\end{table}
	While we were not able to show the tightness of \Cref{th:alphaHG} for other values of $\alpha$, $q$, and $\col$ than the ones described in this section, we found some examples to show the tightness of the result for additional values of $(q,\col)$ when $k=1$ that are not covered by Theorems \ref{thrm:mod_tight}-\ref{theor:ASHG_4} through the use of the integer linear programming approach described in \Cref{sec:method}. \Cref{ta:additional_values} summarizes these results. The smallest case which is unknown (both for fractional and additively separable hedonic games) is the tightness for $q = 5$ and $m = 10$ (see \Cref{ta:additional_values} and \Cref{thrm:mod_tight}). The best lower bound we could find for $q=5$ and $\col=10$ in FHGs is a $(10,1.156)$-stable instance, for example, while \Cref{cor:alphaHG} proves $(10, 1.2)$-stability and \Cref{thrm:mod_tight} shows a lower bound of $(10, 1.125)$. We additionally visualized the tight instances for $q = 5$ and $m \in \{7,8\}$ in Appendix~\ref{app:tight_figures}. A summary of our results can be found in \Cref{ta:summary}, while \Cref{fig:LB_FHG} visualizes the lower bounds for the specific case of FHGs.

	\begin{table*}[t]
		\small
		\caption{Summary and proven tightness results for the values of $f(q,\col)$ such that every $q$-size core stable coalition structure is $(q,f(q,\col))$-stable, with $\col\geq q+1$, for various types of hedonic games. {\no} indicates that tightness is still open. \label{ta:summary}}
		\centering
		\begin{tabularx}{\textwidth}{lc|cccc}
			\toprule
			&&\multicolumn{4}{c}{\textbf{Tightness proof for\ldots}}\\
			\textbf{Hedonic Game} & $\mathbf{f(q,\col)}$ & $(q-1) \mid (m-1)$ & $q=3$ & $q=4$ & Other values of $(q,\col)$\\
			\midrule
			$\alpha$HG & Cor.\ \ref{cor:alphaHG} &\no &\no &\no &$\varnothing$\\
			Hospitable $\alpha$HG& Cor.\ \ref{cor:alphaHG}&Th.\ \ref{thrm:mod_tight}&Th.\ \ref{theorem:LB_3} &\no&$\varnothing$\\
			FHG & $1 + \frac{1}{\col}\left\lfloor \frac{\col-2}{q-1} \right\rfloor$ &Th.\ \ref{thrm:mod_tight} & Th.\ \ref{theorem:LB_3}& Th.\ \ref{theorem:LB_4}&$(q,q+1)$ (Th.\ \ref{theor:q_q+1}) \& \Cref{ta:additional_values}\\
			ASHG & $1+\lfloor\frac{\col-2}{q-1}\rfloor$ &Th.\ \ref{thrm:mod_tight} &Th.\ \ref{theorem:LB_3} &Th.\ \ref{theor:ASHG_4}&$(q,q+1)$ (Th.\ \ref{theor:q_q+1}) \& \Cref{ta:additional_values}\\
			MFHG & 1&Cor.\ \ref{cor:MFSG} & Cor.\ \ref{cor:MFSG} & Cor.\ \ref{cor:MFSG} & \text{all combinations (Cor.\ \ref{cor:MFSG})}\\
			\bottomrule
		\end{tabularx}
	\end{table*}
	
	\begin{figure}
		\centering
		\includegraphics[width=\linewidth]{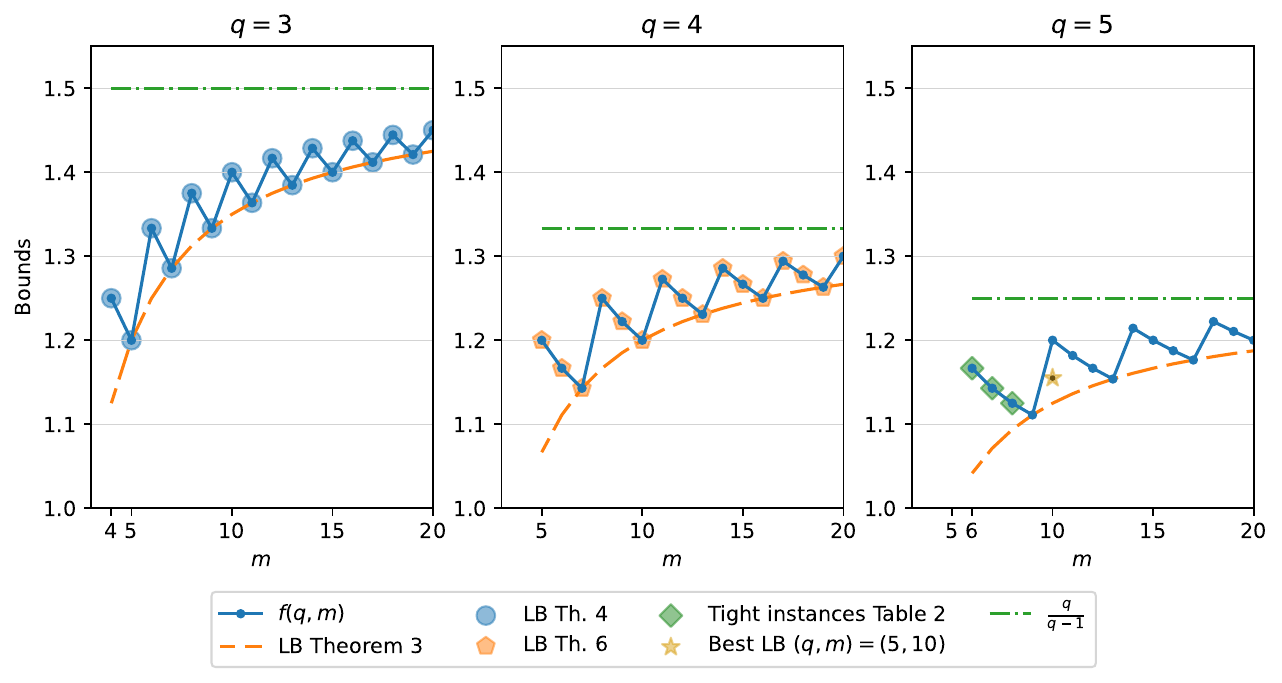}
		\caption{Plotted values of $f(q,\col)$ for FHGs such that every $q$-size core stable coalition structure is $(\col,f(q,\col))$-core stable, together with the obtained lower bounds (LB) in \Cref{sec:lowerBounds}.}
		\label{fig:LB_FHG}
	\end{figure}
	\section{Efficiency}
	\label{sec:efficiency}
	In this section, we study the price of anarchy for core-relaxations of symmetric $\alpha$-hedonic games. Using the same notation as~\citet{FMM21a}, we denote the \emph{social welfare} of a coalition structure $\CC$ by $\text{SW}(\CC) = \sum_{i\in A} u_i(\CC)$, which is simply the sum of the agents' utilities. Moreover, let $\G=(\alpha,A,u)$ represent an instance of an $\alpha$HG. Given an $\alpha$HG $\G$, let $q\text{\textsc{-size Core}}(\G)$ be the set of $q$-size core stable coalition structures, and let $k\text{\textsc{-impr Core}}(\G)$ be the set of $k$-improvement core stable coalition structures. We define the \emph{$q$-size core price of anarchy} of an $\alpha$HG $\G$ as the ratio between the social welfare of the coalition structure $\CC^*(\G)$ that maximizes social welfare, and that of the $q$-size core stable coalition structure with the worst social welfare, i.e., $q\text{\textsc{-size CPoA}}(\G)=\max_{\CC \in q\text{\textsc{-size Core}}(\G)}\frac{SW(\CC^*(\G))}{SW(\CC)}$. Similarly, we define the \emph{$k$-improvement core price of anarchy} as $k\text{\textsc{-impr CPoA}}(\G)=\max_{\CC \in k\text{\textsc{-impr Core}}(\G)}\frac{SW(\CC^*(\G))}{SW(\CC)}$. 
	
	Using \Cref{th:alphaHG}, we can extend the results by~\citet{FMM21a} about the $q$-size core price of anarchy for FHGs.
	\begin{corollary}
		\label{theor:q-CPoA-FHG}
		For any FHG $\G$, it holds that $q\text{\textsc{-size CPoA}}(\G)\leq \frac{2q}{q-1}$, for any integer $q\geq 2$, and this bound is tight.
	\end{corollary}
	\begin{proof}
		By Corollary~\ref{cor:general_FHG}, we know that the social welfare of the worst $q$-size core stable coalition structure is at least the social welfare of the worst $\frac{q}{q-1}$-improvement core stable coalition structure. Moreover, by Theorem~8 by \cite{FMM21a}, we know that the $k$-improvement CPoA of a FHG is upper bounded by $2k$, for any $k\geq 1$. The tightness of the bound follows from Theorem 9 by \cite{FMM21a}.
	\end{proof}
	
	Next, we show upper bounds on the $q\text{\textsc{-size CPoA}}(\G)$ and the $k\text{\textsc{-impr CPoA}}(\G)$ for the following subclass of $\alpha$HGs.
	\begin{definition}
		A function $\alpha \colon [n] \to \mathbb{R}$ is \emph{decreasing} if $\alpha_{q}\geq \alpha_{q+1}$ for all integers $q\geq 1$. Accordingly, an $\alpha$HG is decreasing if $\alpha$ is decreasing.
	\end{definition}
	
	The class of decreasing $\alpha$HGs is distinct from the class of hospitable $\alpha$HGs defined in Definition~\ref{def:hospitable}, but it also contains FHGs, ASHGs, and MFHGs. When restricting our focus to the subclass of decreasing $\alpha$HGs, we can generalize Theorem 8 by \cite{FMM21a} to obtain the same upper bound on the $k$-improvement core price of anarchy.
	\begin{theorem}
		\label{theor:PoA_alpha_improv}
		For any decreasing $\alpha$HG $\G$ and for every $k\geq1$, $k\text{\textsc{-impr CPoA}}(\G)\leq 2k$.
	\end{theorem}
	\begin{proof}
		The proof is identical to the proof of Theorem 8 by \citet{FMM21a}. Using their notation, with the adapted definition that $\mu^>_{i}(C)=\alpha(|C|)\cdot\delta^>_{C}(i)$, the only required alteration to their proof is that equation (2) in their proof should be replaced by:
		\begin{equation*}
			\mu_{i_t}(C^*_t) = \alpha(|C^*_t|)\cdot\delta^>_{C^*}(i_t)\geq\alpha(|C^*|)\cdot\delta^>_{C^*}(i_t) = \mu_{i_t}^>(C^*),
		\end{equation*}
		which holds, by definition, because we are only considering decreasing $\alpha$HGs.
	\end{proof}
	
	Lastly, we can use a similar reasoning as in the proof of \Cref{theor:q-CPoA-FHG} and use the bounds from Theorems \ref{th:alphaHG} and \ref{theor:PoA_alpha_improv} to obtain a general upper bound on the $q$-size core price of anarchy for decreasing $\alpha$HGs.
	\begin{theorem}
		\label{theor:PoA_alpha_size}
		For any decreasing $\alpha$HG $\G$, $q\text{\textsc{-size CPoA}}(\G) \le 2\cdot\max_{m\geq q+1}f(m,q)$.
	\end{theorem}
	Note that this result implies a core price of anarchy of 2 for MFHGs, where the core price of anarchy of an $\alpha$HG $\G$ is simply defined as $\max_q q\text{\textsc{-size CPoA}}(\G)$. As such, we answered an open question by \citet{MMV20a}, who found a lower bound on the core price of stability of 2 and an upper bound for the core price of anarchy of 4 in MFHGs.
	
	\begin{corollary}
		For any MFHG, the core price of anarchy is upper bounded by 2.
	\end{corollary}
	
	\section{Applications for Novel Hedonic Games Models}
	\label{sec:novel}
	To illustrate the usefulness of our main result, we will introduce two hedonic game variants that can be modeled as an $\alpha$HG which have been unstudied to the best of our knowledge. We will show how our result can be used to easily prove the existence of core stable coalition structures.
	
	\subsection{Pairwise Communication}
	\label{subsec:pairwise_comm}
	Consider a setting where members in a coalition have to communicate extensively with each other individual member of the coalition. When all form of communication within a coalition is costly for each member of the coalition, one could model the individual utility of belonging to a coalition to be proportional to the total amount of pairwise communication in that coalition. Note the difference of this setting with MFHGs, in which agents only care about their personal amount of pairwise communication. Such considerations might arise, for example, in governmental coalition formation. When multiple political parties are discussing the formation of a government, pairwise communication between parties is important  to discuss strategies and compromises. Larger coalitions will require more pairwise alignment between parties, causing the coalition formation process to advance slower and to have a larger risk of failing.
	
	As the total number of edges in a complete graph of size~$\col$ equals $\frac{\col(\col-1)}{2}$, we can model this setting as an $\alpha$HG with $\alpha_\col=\frac{2}{\col(\col-1)}$, where $\alpha_1=1$. Note that this is a hospitable $\alpha$HG. 
	
	Our results directly imply the existence of a core stable coalition structure, and an upper bound on the price of anarchy of 2.
	
	\begin{corollary}
		\label{cor:application}
		A core-stable coalition structure always exists in $\frac{2}{\col(\col-1)}$HGs.
	\end{corollary}
	\begin{proof}
		By \Cref{cor:core_char} it is sufficient to show that $\frac{(m-1) \alpha_m}{\alpha_2} \le 1$ for all $m \ge 3$. As $\frac{(m-1) \alpha_m}{\alpha_2} = (m-1) \frac{2}{m (m-1)} = \frac{2}{m} \le \frac{2}{3}$, this always holds. 
	\end{proof}
	
	\begin{corollary}
		For any $\frac{2}{\col(\col-1)}$HG, the core price of anarchy is upper bounded by 2.
	\end{corollary}
	
	\subsection{Odd vs.\ Even}
	\label{subsec:odd_even}
	As a second possible hedonic game consider the following: in each coalition, a maximal matching is selected uniformly at random (because the agents should work in pairs, for example). The utility of an agent is their expected utility of their matching partner (with an agent getting a utility of $0$ if they are unmatched). This equivalently leads to an $\alpha$HG with
	\begin{equation}
		\alpha(m) = \begin{cases}
			\frac{1}{m-1} &\text{if $m$ is even,}\\
			\frac{1}{m} &\text{if $m$ is odd.}
			\label{eq:odd_even}
		\end{cases}
	\end{equation}
	In other words, $\alpha(m)$ is the same as in MFHGs when $\col$ is even, and the same as in FHGs when $\col$ is uneven. 
	
	Note that this is not a hospitable $\alpha$HG (\eg $\frac{\alpha_5}{\alpha_4}=0.6<\frac{3}{4}$). In fact, these hedonic games show that the bound derived in \Cref{cor:alphaHG} can be non tight for non-hospitable $\alpha$HGs. To see this, we first observe that any $2$-size core stable outcome is core stable following \Cref{cor:core_char}. However, for $4$-size stable coalition structures, \Cref{cor:alphaHG} only implies $(6,1.2)$-core stability (and not ``pure'' core stability). 
	
		
		\section{Conclusion and Outlook}
		\label{sec:conclusion}
		We studied hedonic games and the relationship between different relaxed notions of core stability. Most importantly, for a large class of hedonic games, we obtained a general upper bound $f(q,m)$ such that every $q$-size core stable outcome is $(m, f(q,m))$-core stable. That is, a coalition of size $m$ can deviate at most by a factor of $f(q,m)$. Our bound also allows us to answer a conjecture by \citet{FMM21a} that every $q$-size core stable outcome in symmetric fractional hedonic games is $\frac{q}{q-1}$-improvement core stable. Finally, we also obtain some lower bounds. However, even for fractional and additively separable hedonic games, our bounds are not tight yet. For both, we were only able to show the tightness up to $q = 4$, with the smallest open case being $q=5$ and $\col=10$. 
		For both kinds of hedonic games, our integer linear programming approach was not able to construct a counterexample, nor show that no counterexample exists. 
		
		Improving our lower bounds, not only for fractional and additively separable hedonic games, but also for the general class of hospitable hedonic games, seems like a challenging and interesting venue for future work. Further, it would be interesting to see if the generalization of $\alpha$-hedonic games could see application in other areas of hedonic games as well. For instance, it might be interesting to classify for which $\alpha$-hedonic games certain dynamics converge (see for instance \citet{BBK23a}), to generalize recent algorithmic results from ASHGs \citep{BuRo23a, BuRo24a}, or to extend results on single-agent stability \citep{BBT24a, ABS13a, SuDi10a} concepts to $\alpha$-hedonic games.

		\setlength{\parskip}{10pt}
		\footnotesize \textbf{Acknowledgements} Tom Demeulemeester is funded by PhD fellowship 11J8721N of Research Foundation - Flanders (FWO) and by the Swiss National Science Foundation (SNSF) through Project 100018$\_$212311. Jannik Peters was supported by Deutsche Forschungsgemeinschaft (DFG) under the grant BR 4744/2-1 and the Graduiertenkolleg “Facets of Complexity” (GRK 2434) as well as by the Singapore Ministry of Education under grant
		number MOE-T2EP20221-0001.
		\setlength{\parskip}{0pt}
		
		\normalsize
		\bibliographystyle{apalike} 
		\bibliography{abb, algo}

\begin{thebibliography}{}

\bibitem[Aamand et~al., 2023]{ACL+23a}
Aamand, A., Chen, J., Liu, A., Silwal, S., Sukprasert, P., and Vakilian, A.
  (2023).
\newblock Constant approximation for individual preference stable clustering.
\newblock In {\em Proceedings of the 37th Conference on Neural Information
  Processing Systems (NeurIPS)}, pages 43646--43661.

\bibitem[Ahmadi et~al., 2022]{AAK+22a}
Ahmadi, S., Awasthi, P., Khuller, S., Kleindessner, M., Morgenstern, J.,
  Sukprasert, P., and Vakilian, A. (2022).
\newblock Individual preference stability for clustering.
\newblock In {\em Proceedings of the 39th International Conference on Machine
  Learning (ICML)}, pages 162:197--246.

\bibitem[Alcalde-Unzu et~al., 2024]{AGIM24a}
Alcalde-Unzu, J., Gallo, O., Inarra, E., and Moreno-Ternero, J.~D. (2024).
\newblock Solidarity to achieve stability.
\newblock {\em European Journal of Operational Research}, 315(1):368--377.

\bibitem[Aziz et~al., 2019]{ABB+19a}
Aziz, H., Brandl, F., Brandt, F., Harrenstein, P., Olsen, M., and Peters, D.
  (2019).
\newblock Fractional hedonic games.
\newblock {\em ACM Transactions on Economics and Computation}, 7(2):1--29.

\bibitem[Aziz et~al., 2013]{ABS13a}
Aziz, H., Brandt, F., and Stursberg, P. (2013).
\newblock On popular random assignments.
\newblock In {\em Proceedings of the 6th International Symposium on Algorithmic
  Game Theory (SAGT)}, volume 8146 of {\em Lecture Notes in Computer Science
  (LNCS)}, pages 183--194. Springer-Verlag.

\bibitem[Aziz et~al., 2015]{AGG+15b}
Aziz, H., Gaspers, S., Gudmundsson, J., Mestre, J., and T{\"a}ubig, H. (2015).
\newblock Welfare maximization in fractional hedonic games.
\newblock In {\em Proceedings of the 24th International Joint Conference on
  Artificial Intelligence (IJCAI)}, pages 461--467. AAAI Press.

\bibitem[Aziz and Savani, 2016]{AzSa15a}
Aziz, H. and Savani, R. (2016).
\newblock Hedonic games.
\newblock In Brandt, F., Conitzer, V., Endriss, U., Lang, J., and Procaccia,
  A.~D., editors, {\em Handbook of Computational Social Choice}, chapter~15.
  Cambridge University Press.

\bibitem[Bil{\`o} et~al., 2018]{BFF+18a}
Bil{\`o}, V., Fanelli, A., Flammini, M., Monaco, G., and Moscardelli, L.
  (2018).
\newblock Nash stable outcomes in fractional hedonic games: {E}xistence,
  efficiency and computation.
\newblock {\em Journal of Artificial Intelligence Research}, 62(315--371).

\bibitem[Boehmer et~al., 2023]{BBK23a}
Boehmer, N., Bullinger, M., and Kerkmann, A. (2023).
\newblock Causes of stability in dynamic coalition formation.
\newblock In {\em Proceedings of the 37th AAAI Conference on Artificial
  Intelligence (AAAI)}, pages 5499--5506.

\bibitem[Boehmer and Elkind, 2020]{BoEl20a}
Boehmer, N. and Elkind, E. (2020).
\newblock Stable roommate problem with diversity preferences.
\newblock In {\em Proceedings of the 29th International Joint Conference on
  Artificial Intelligence (IJCAI)}, pages 96--102. ijcai.org.

\bibitem[Bogomolnaia and Jackson, 2002]{BoJa02a}
Bogomolnaia, A. and Jackson, M.~O. (2002).
\newblock The stability of hedonic coalition structures.
\newblock {\em Games and Economic Behavior}, 38(2):201--230.

\bibitem[Brandt et~al., 2024]{BBT24a}
Brandt, F., Bullinger, M., and Tappe, L. (2024).
\newblock Stability based on single-agent deviations in additively separable
  hedonic games.
\newblock {\em Artificial Intelligence}, 334:104160.

\bibitem[Bredereck et~al., 2019]{BEI19a}
Bredereck, R., Elkind, E., and Igarashi, A. (2019).
\newblock Hedonic diversity games.
\newblock In {\em Proceedings of the 18th International Conference on
  Autonomous Agents and Multiagent Systems (AAMAS)}, pages 565--573.

\bibitem[Bullinger and Kober, 2021]{BuKo21a}
Bullinger, M. and Kober, S. (2021).
\newblock Loyalty in cardinal hedonic games.
\newblock In {\em Proceedings of the 30th International Joint Conference on
  Artificial Intelligence (IJCAI)}, pages 66--72.

\bibitem[Bullinger and Romen, 2023]{BuRo23a}
Bullinger, M. and Romen, R. (2023).
\newblock Online coalition formation under random arrival or coalition
  dissolution.
\newblock In {\em Proceedings of the 31st European Symposium on Algorithms
  (ESA)}, pages 27:1--27:18.

\bibitem[Bullinger and Romen, 2024]{BuRo24a}
Bullinger, M. and Romen, R. (2024).
\newblock Stability in online coalition formation.
\newblock In {\em Proceedings of the 38th AAAI Conference on Artificial
  Intelligence (AAAI)}, pages 9537--9545.

\bibitem[Caragiannis et~al., 2024]{CMS24b}
Caragiannis, I., Micha, E., and Shah, N. (2024).
\newblock Proportional fairness in non-centroid clustering.
\newblock In {\em Proceedings of the 38th Conference on Neural Information
  Processing Systems (NeurIPS)}.
\newblock Forthcoming.

\bibitem[Carosi et~al., 2019]{CMM19a}
Carosi, R., Monaco, G., and Moscardelli, L. (2019).
\newblock Local core stability in simple symmetric fractional hedonic games.
\newblock In {\em Proceedings of the 18th International Conference on
  Autonomous Agents and Multiagent Systems (AAMAS)}, pages 574--582.

\bibitem[Cheng et~al., 2023]{CPM23a}
Cheng, X., Pan, C., and Maghsudi, S. (2023).
\newblock Parallel online clustering of bandits via hedonic game.
\newblock In {\em Proceedings of the 40th International Conference on Machine
  Learning (ICML)}, pages 5485--5503.

\bibitem[Cheng et~al., 2020]{CJMW20a}
Cheng, Y., Jiang, Z., Munagala, K., and Wang, K. (2020).
\newblock Group {{Fairness}} in {{Committee Selection}}.
\newblock {\em ACM Transactions on Economics and Computation},
  8(4):{23:1--23:18}.

\bibitem[Dr{\`e}ze and Greenberg, 1980]{DrGr80a}
Dr{\`e}ze, J.~H. and Greenberg, J. (1980).
\newblock Hedonic coalitions: Optimality and stability.
\newblock {\em Econometrica}, 48(4):987--1003.

\bibitem[Fanelli et~al., 2021]{FMM21a}
Fanelli, A., Monaco, G., and Moscardelli, L. (2021).
\newblock Relaxed core stability in fractional hedonic games.
\newblock In {\em Proceedings of the 30th International Joint Conference on
  Artificial Intelligence (IJCAI)}, pages 182---188.

\bibitem[Gairing and Savani, 2019]{GaSa19a}
Gairing, M. and Savani, R. (2019).
\newblock Computing stable outcomes in symmetric additively separable hedonic
  games.
\newblock {\em Mathematics of Operations Research}, 44(3):1101--1121.

\bibitem[Ganian et~al., 2023]{GHK+23a}
Ganian, R., Hamm, T., Knop, D., Schierreich, S., and Such{\'{y}}, O. (2023).
\newblock Hedonic diversity games: {A} complexity picture with more than two
  colors.
\newblock {\em Artificial Intelligence}, 325:104017.

\bibitem[Monaco et~al., 2020]{MMV20a}
Monaco, G., Moscardelli, L., and Velaj, Y. (2020).
\newblock Stable outcomes in modified fractional hedonic games.
\newblock {\em Autonomous Agents and Multi-Agent Systems}, 34(1):1--29.

\bibitem[Olsen, 2012]{Olse12a}
Olsen, M. (2012).
\newblock On defining and computing communities.
\newblock In {\em Proceedings of the 18th Computing: {Australasian} Theory
  Symposium (CATS)}, volume 128 of {\em Conferences in Research and Practice in
  Information Technology (CRPIT)}, pages 97--102.

\bibitem[Peters and Elkind, 2015]{PeEl15a}
Peters, D. and Elkind, E. (2015).
\newblock Simple causes of complexity in hedonic games.
\newblock In {\em Proceedings of the 25th International Joint Conference on
  Artificial Intelligence (IJCAI)}, pages 617--623.

\bibitem[Sung and Dimitrov, 2010]{SuDi10a}
Sung, S.~C. and Dimitrov, D. (2010).
\newblock Computational complexity in additive hedonic games.
\newblock {\em European Journal of Operational Research}, 203(3):635--639.

\end{thebibliography}
		
		\appendix
		
		\section{Proof of \texorpdfstring{\Cref{theorem:LB_4}}{Theorem 6}}
		\label{ap:proof_FHG_4}
		
		\begin{proof}
			Assume we are given $\col$ agents $a_1, \dots, a_\col$ with $u_i(\CC) = 1$. We partition the set of agents into two sets of agents. We call the agents $a_1, \dots, a_{\left\lfloor \frac{\col-2}{3}\right\rfloor + 1}$
			the \emph{two-valued} agents and the agents
			$a_{\left\lfloor \frac{\col-2}{3}\right\rfloor + 2}, \dots, a_\col$
			the \emph{one-valued} agents. 
			For each \begin{itemize}
				\item $i,j \in \left[\left\lfloor \frac{\col-2}{3}\right\rfloor + 1\right]$ we set $u(i,j) = 0$ ;
				\item $i\in \left[\left\lfloor \frac{\col-2}{3}\right\rfloor + 1\right], j \in \left[\left\lfloor \frac{\col-2}{3}\right\rfloor + 2, \col\right]$ we set $u(i,j) = 2$;
				\item $i,j \in \left[\left\lfloor \frac{\col-2}{3}\right\rfloor + 2, \col\right]$ we set $u(i,j) = 1$.
			\end{itemize}
			First, we see that this coalition is 2-size stable. To show that it is 3-size stable, we distinguish four cases and show that in each case there is an agent with a utility of at most $1$ and thus no blocking coalition can be formed. Let $C = \{a_i,a_j,a_k\}$, 
			\begin{itemize}
				\item If all three agents are two-valued agents, they all have utility 0.
				\item If two of the agents are two-valued agents, they have utility $\frac{2}{3}$.
				\item If only one agent is a two-valued agent, the one-valued agents have utility $1$.
				\item If none of the agents are two valued agents, each of the agents has utility $\frac{2}{3}$.
			\end{itemize}
			A similar reasoning holds for coalitions of size four:
			\begin{itemize}
				\item If all four agents are two-valued agents, they all have utility 0.
				\item If three are two-valued agents, they have utility  $\frac{1}{2}$.
				\item If two of the agents are two-valued agents, they have utility $1$.
				\item If only one agent is a two-valued agent, the one-valued agents have utility $\frac{2 + 1 + 1}{4} = 1$
				\item If none of the agents are two valued agents, each of the agents has utility $\frac{3}{4}$.
			\end{itemize}
			Hence, this coalition structure is $4$-size stable. Further, we see that in the coalition compromised of all agents, all two-valued agents get utility
			\begin{align*}
				\frac{2 (\col - \left\lfloor \frac{\col-2}{3}\right\rfloor - 1)}{\col} \ge 1 + \frac{\left\lfloor\frac{1}{3}(\col-2)\right\rfloor}{\col}.
			\end{align*} 
			Finally, any one-valued agent gets a utility of 
			\begin{align*}
				\frac{2 (\left\lfloor \frac{\col-2}{3}\right\rfloor + 1 ) + (\col - \left\lfloor \frac{\col-2}{3}\right\rfloor - 2) }{\col} = 1 + \frac{\left\lfloor\frac{1}{3}(\col-2)\right\rfloor}{\col}
			\end{align*} 
			Hence, a $4$-size stable coalition is not $\left(\col, 1+\frac{\left\lfloor\frac{1}{3}(\col-2)\right\rfloor}{\col}-\varepsilon\right)$-stable for any $\varepsilon>0$.
		\end{proof}

		\section{Proof of \texorpdfstring{\Cref{theor:ASHG_4}}{Theorem 7}}
		\label{ap:proof_ASHG_4}
		\begin{proof}
			To prove this result, we propose three different constructions, depending on the value of $\col\bmod 3$. First, when $\col\bmod 3 = 1$, the result follows from \Cref{thrm:mod_tight}.
			
			Second, when $\col \bmod 3 = 0$, we get that $1 + \lfloor \frac{m-2}{3} \rfloor = \frac{m}{3}$. For our construction, we divide the agents into two groups $A_1, A_2$ with sizes $\frac{2m}{3}$ and $\frac{m}{3}$, respectively. Let the utilities of the agents be equal to $u_i(\CC) = 1$ for any $a_i \in A_1$ and $u_i(\CC) = 2$ for any $a_i \in A_2$. Further, let $u(i,j) = 1$ for any $i \in A_1$ and $j \in A_2$, and let $u(k,\ell)=0$ for all other edges.
			
			This structure is obviously $2$-size stable. For $3$-size stability we observe that any blocking coalition of size $3$ could not contain an agent from $A_2$ since they could at most get a utility of $2$. Hence, it should contain three agents from $A_1$ who get a utility of $0$, and thus do not form a blocking coalition either.
			
			Moreover, we see that a coalition is clearly not $4$-size blocking when it contains four agents from $A_1$, or when it contains three or four agents from $A_2$. In a coalition containing three agents from $A_1$, the utility of the agents in $A_1$ equals 1, and hence this coalition is not blocking. Moreover, in a coalition containing two agents from $A_1$, the utility of the agents in $A_2$ equals 2, which implies that this coalition is not blocking either.
			
			Finally, we observe that in $A_1 \cup A_2$ the agents in $A_1$ experience a utility of $\frac{m}{3}$ and the agents in $A_2$ a utility of $\frac{2m}{3}$, as desired.
			
			Third, when $\col\bmod 3 = 2$, we observe that $1 + \lfloor \frac{m-2}{3}\rfloor = 1 + \frac{m-2}{3}$. Now we consider a coalition structure~$\CC$ and a coalition~$C$ of size~$\col$ that contains two types of agents: $A_1 = \{a_1, \ldots, a_{\frac{\col-2}{3}}\}$, and $A_2 = \{a_{\frac{\col-2}{3}+1},\ldots, a_{\col}\}$. Let $u_i(\CC) = 2$ for $a_i\in A_1$ and $u_j(\CC) = 1$ for $a_j\in A_2$. Moreover, let $u(i,j) = 1$ for $\{a_i,a_j\}\in A_1 \times A_2$. Further, for any $i \in [\frac{(\col-2)}{3} + 1]$, we set $u(\frac{\col-2}{3}+2i-1,\frac{\col-2}{3}+2i) = 1$. Hence, every agent in $A_2$ has exactly one other agent in $A_2$ they give a utility of $1$ to. We again observe that this coalition structure is 2-size stable. Any $3$-size blocking coalition cannot contain an agent from $A_1$, which in turn implies that the agents in $A_2$ can get a utility of at most $1$ from the blocking coalition. 
			
			Lastly, any coalition of size four which contains three or four agents from $A_1$ or which contains four agents from $A_2$ is clearly not $4$-size blocking. In a coalition containing two agents from $A_1$, those agents experience a utility of 2, which implies that this coalition is not blocking. Moreover, in a coalition containing one agent from $A_1$, there is at least one agent in $A_2$ who experiences a utility of at most 1, and hence this coalition is not blocking either. 
			
			Finally, we observe that every agent in $A_1$ gets a utility of $2(1 + \frac{m-2}{3})$ and every agent in $A_2$ a utility of $1 + \frac{m-2}{3}$.
		\end{proof}
		
		\section{Tight Examples}\label{app:tight_figures}
		Figures \ref{fig:tight1} and \ref{fig:tight2} contain tight examples for FHGs when $q=5$ and $m\in\{7,8\}$, and Figures \ref{fig:tight3} and\ref{fig:tight4} contain tight examples for ASHGs when  $q=5$ and $m\in\{7,8\}$.
		%
		%

		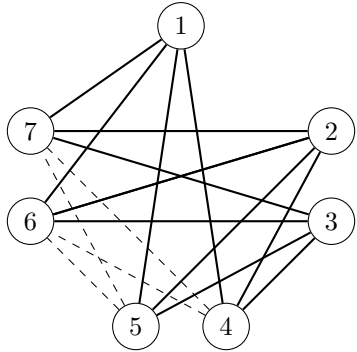
\begin{figure}[h]
			\centering
			\begin{tikzpicture}
				\node[circle,draw] (1) at (2,4) {1};
				\node[circle,draw] (6) at (0,1.4) {6};
				\node[circle,draw] (2) at (4,2.6) {2};
				\node[circle,draw] (5) at (1.4,0) {5};
				\node[circle,draw] (4) at (2.6,0) {4};
				\node[circle,draw] (3) at (4,1.4) {3};
				\node[circle,draw] (7) at (0,2.6) {7};
				
				\draw[thick] (1) -- (7);
				\draw[thick] (1) -- (6);
				\draw[thick] (1) -- (5);
				\draw[thick] (1) -- (4);
				\draw[thick] (2) -- (7);
				\draw[thick] (2) -- (6);
				\draw[thick] (2) -- (5);
				\draw[thick] (2) -- (4);
				\draw[thick] (3) -- (7);
				\draw[thick] (3) -- (6);
				\draw[thick] (3) -- (5);
				\draw[thick] (3) -- (4);
				\draw[thick] (2) -- (6);
				
				\draw[dashed] (4) -- (7);
				\draw[dashed] (4) -- (6);
				\draw[dashed] (5) -- (7);
				\draw[dashed] (5) -- (6);
			\end{tikzpicture}
			\caption{Tight example for FHGs for $q = 5$ and $m = 7$. Solid edges indicate a weight of $2$ while dashed edges indicate a weight of $1$. Each agent has an initial utility of $1$.}
			\label{fig:tight1}
			
		\end{figure}
		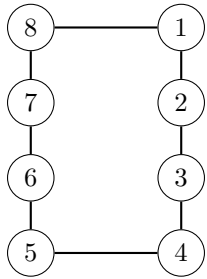
\begin{figure}[h]
			\centering
			\begin{tikzpicture}
				\node[circle,draw] (1) at (0,4) {8};
				\node[circle,draw] (2) at (0,3) {7};
				\node[circle,draw] (3) at (2,3) {2};
				\node[circle,draw] (4) at (0,2) {6};
				\node[circle,draw] (5) at (0,1) {5};
				\node[circle,draw] (6) at (2,4) {1};
				\node[circle,draw] (7) at (2,1) {4};
				\node[circle,draw] (8) at (2,2) {3};
				
				\draw[thick] (1) -- (2);
				\draw[thick] (1) -- (6);
				\draw[thick] (2) -- (4);
				\draw[thick] (3) -- (6);
				\draw[thick] (3) -- (8);
				\draw[thick] (4) -- (5);
				\draw[thick] (5) -- (7);
				\draw[thick] (7) -- (8);
				
			\end{tikzpicture}
			\caption{Tight example for FHGs for $q = 5$ and $m = 8$. Solid edges indicate a weight of $2$ while no edge indicates a weight of~$1$. Each agent has an initial utility of $1$.}
			\label{fig:tight2}
		\end{figure}

		\begin{figure}[H]
			\centering
			\begin{tikzpicture}
				\node[circle,draw] (2) at (2,4) {2};
				\node[circle,draw] (7) at (0,1.2) {7};
				\node[circle,draw] (3) at (4,2.8) {3};
				\node[circle,draw] (6) at (1.2,0) {6};
				\node[circle,draw] (5) at (2.8,0) {5};
				\node[circle,draw] (4) at (4,1.2) {4};
				\node[circle,draw] (1) at (0,2.8) {1};
				
				\draw[thick] (1) -- (2);
				\draw[thick] (2) -- (3);
				\draw[thick] (7) -- (1);
				
				\draw[dashed] (3) -- (4);
				\draw[dashed] (3) -- (5);
				\draw[dashed] (4) -- (5);
				\draw[dashed] (5) -- (6);
				\draw[dashed] (6) -- (7);
			\end{tikzpicture}
			\caption{Tight example for ASHGs for $q = 5$ and $m = 7$. Solid edges indicate a weight of $2$ while dashed edges indicate a weight of $1$. Agents $1,2$, and $3$ have an initial utility of $2$ and agents $4,5,6$, and $7$ have an initial utility of $1$.}
			\label{fig:tight3}
			
		\end{figure}
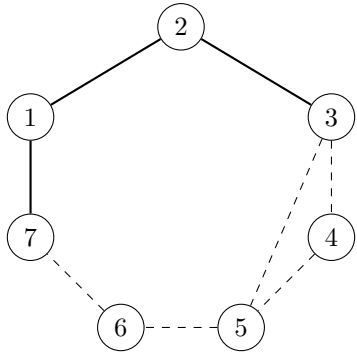

		\begin{figure}[H]
			\centering
			\begin{tikzpicture}
				\node[circle,draw] (1) at (2,4) {1};
				\node[circle,draw] (6) at (0,1.2) {7};
				\node[circle,draw] (2) at (4,2.8) {2};
				\node[circle,draw] (5) at (4,1.2) {4};
				\node[circle,draw] (4) at (2.8,0) {5};
				\node[circle,draw] (3) at (2,2) {3};
				\node[circle,draw] (7) at (0,2.8) {8};
				\node[circle,draw] (8) at (1.2,0) {6};
				
				\draw[thick] (1) -- (2);
				
				\draw[dashed] (1) -- (3);
				\draw[dashed] (1) -- (7);
				\draw[dashed] (2) -- (4);
				\draw[dashed] (2) -- (5);
				\draw[dashed] (3) -- (5);
				\draw[dashed] (3) -- (6);
				\draw[dashed] (3) -- (8);
				\draw[dashed] (4) -- (8);
				\draw[dashed] (7) -- (6);
				
			\end{tikzpicture}
			\caption{Tight example for ASHGs for $q = 5$ and $m = 8$. Solid edges indicate a weight of $2$ while dashed edges indicate a weight of $1$. Agents $1,2$, and $3$ have an initial utility of $2$ and agents $4,5,6,7$, and $8$ have an initial utility of $1$.}
			\label{fig:tight4}
		\end{figure}
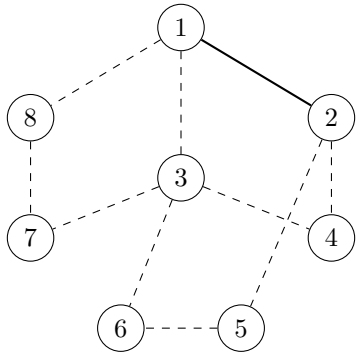

		\vspace*{\fill}

	\end{document}